\documentclass[aps,twocolumn,amsmath,amssymb,pra,superscriptaddress]{revtex4-2}

\usepackage{svg,times,upgreek}
\usepackage{color}
\usepackage{graphicx}% Include figure files
\usepackage{dcolumn}% Align table columns on decimal point
\usepackage{bm}% bold math
\usepackage{physics}
\usepackage{amsmath}
\usepackage{notes2bib}
\usepackage{ dsfont }
\usepackage{textcomp}
\usepackage{svg}
\usepackage{verbatim}
\usepackage{amssymb}
\usepackage{amsfonts}              % for blackboard bold, etc
\usepackage{amsthm}                % better theorem environments
\usepackage{mathtools}
\usepackage{kotex}
\usepackage[normalem]{ulem}
\usepackage{xcolor}
\usepackage{tikz}
\usetikzlibrary{quantikz}

\newcommand{\shorteq}{%
  \settowidth{\@tempdima}{-}% Width of hyphen
  \resizebox{\@tempdima}{\height}{=}%
}
\newcommand{\D}{\mathrm{d}}
\newcommand{\I}{\mathrm{i}}
\newcommand{\RE}[1]{\mathrm{Re}\{#1\}}
\newcommand{\IM}[1]{\mathrm{Im}\{#1\}}

\DeclarePairedDelimiterX{\barpair}[2]{(}{)}{%
  #1\;\delimsize\|\;#2%
}

\newtheorem{theorem}{Theorem}

\newcommand{\opt}{{p_\mathrm{hack}^\mathrm{opt}}}
\newcommand{\PG}{{p_\mathrm{hack}^\textsc{PG}}}
\newcommand{\ME}{{p_\mathrm{hack}^\textsc{ME}}}

\bibnotesetup{note-name    = ,use-sort-key = false} %Removes 'Note1' handle for note2bib commands
\begin{document}

\title{Hacking Quantum Networks: Extraction and Installation of Quantum Data}

\author{Seok Hyung Lie}
\author{Yong Siah Teo}
\author{Hyunseok Jeong}\email{h.jeong37@gmail.com}
\affiliation{%
 Department of Physics and Astronomy, Seoul National University, Seoul, 151-742, Korea
}%

\date{\today}

\begin{abstract}
We study the problem of quantum hacking, which is the procedure of quantum-information extraction from and installation on a quantum network given only partial access. This problem generalizes a central topic in contemporary physics---information recovery from systems undergoing scrambling dynamics, such as the Hayden--Preskill protocol in black-hole studies. We show that a properly prepared partially entangled probe state can generally outperform a maximally entangled one in quantum hacking. Moreover, we prove that finding an optimal decoder for this stronger task is equivalent to that for Hayden--Preskill-type protocols, and supply analytical formulas for the optimal hacking fidelity of large networks. In the two-user scenario where Bob attempts to hack Alice's data, we find that the optimal fidelity increases with Bob's hacking space relative to Alice's user space. However, if a third neutral party, Charlie, is accessing the computer concurrently, the optimal hacking fidelity against Alice drops with Charlie's user-space dimension, rendering targeted quantum hacking futile in high-dimensional multi-user scenarios without classical attacks. When applied to the black-hole information problem, the limited hacking fidelity implies a reflectivity decay of a black hole as an information mirror. 
\end{abstract}

\pacs{Valid PACS appear here}
\maketitle

Advancements in quantum technologies have triggered anticipations of a new quantum-computing era, where each quantum computer can be connected either to a quantum network \cite{yurke1984quantum, elliott2002building, simon2017towards}, or to the quantum internet \cite{kimble2008quantum, wehner2018quantum, pirandola2016physics, caleffi2018quantum, castelvecchi2018quantum}. In such a delocalized quantum-computation setting, a third party can be granted access to a part of a quantum network. Suppose that Alice operates such a network with an architecture that is publicly given by a multipartite unitary operator. Alice grants Bob, a client, access to an input port of the network. In this situation, Bob might be a hacker who attempts to extract as much data from Alice as possible.

This problem can naturally emerge in non-artificial quantum interactions. It reduces to the information recovery problem, which is actively studied in the field of scrambling of quantum information \cite{sekino2008fast,hosur2016chaos,hosur2016chaos,blok2021scrambling}. One notable example is the Hayden--Preskill protocol for recovering quantum information from the Hawking radiation of old black holes \cite{page1994black, hayden2007black, bao2021hayden, cheng2020realizing, yoshida2019firewalls}. Its optimal decoding map is not completely understood and requires creative construction. In~\cite{yoshida2017efficient}, although an efficient decoding map was proposed for the information recovery task, its optimality is still an open question. On the other hand, from the perspective of the decoupling approach \cite{dupuis2010decoupling} to quantum information, perfect recovery is equivalent to implementing a quantum channel without leaking information to a particular subsystem. This is also studied in the context of catalysis of quantum randomness \cite{boes2018catalytic, wilming2020entropy, lie2019unconditionally, lie2020randomness, lie2020uniform}.

By the no-cloning theorem~\cite{wootters1982single}, however, Bob cannot simply copy out Alice's quantum data, but has to replace it with another. Although this aspect of information recovery is generally overlooked, the problem of deciding which data to install and how well this installation can be implemented is now relevant. If Bob additionally wants to extract information of Alice's next computation, he needs the side information of the quantum state currently in Alice's memory. Naturally, the ability to change the quantum data of Alice to whatever Bob prepared is the most desirable, but we will show that it is possible only in trivial cases. Therefore, when Alice's next step is unknown, Bob would typically want to replace Alice's quantum data with a part of a maximally entangled state whose reference system is in his possession for the maximal side information. We refer to this task of extraction and installation of quantum data as \textit{quantum hacking}. 

{\it Quantum hacking}.---For simplicity, we first consider a quantum network accessed by only two users, Alice and Bob, described as a $d_Ad_B$-dimensional unitary operator $U$ on systems $AB$. As a hacker, Bob might try to extract as much quantum information stored in $A$ as possible and replace it with an arbitrary quantum state of his choice. This is equivalent to feeding a `probe' quantum state that depends on the quantum data Bob wants to install into $U$ and applying a recovery map to the output state to simulate the \texttt{SWAP} operation \cite{garcia2013swap}. However, if this task can be done with error $\epsilon$ (measured by the average infidelity of pure state inputs), then the unitary operator $U$ itself should be close to the \texttt{SWAP} operator (up to local operations) with error $\epsilon$ and vice versa. It is thus impossible to substitute quantum data through a non-\texttt{SWAP} operator~\cite{SM_Qhack}.

\begin{figure}
	\centering
	\includegraphics[width=0.8\columnwidth]{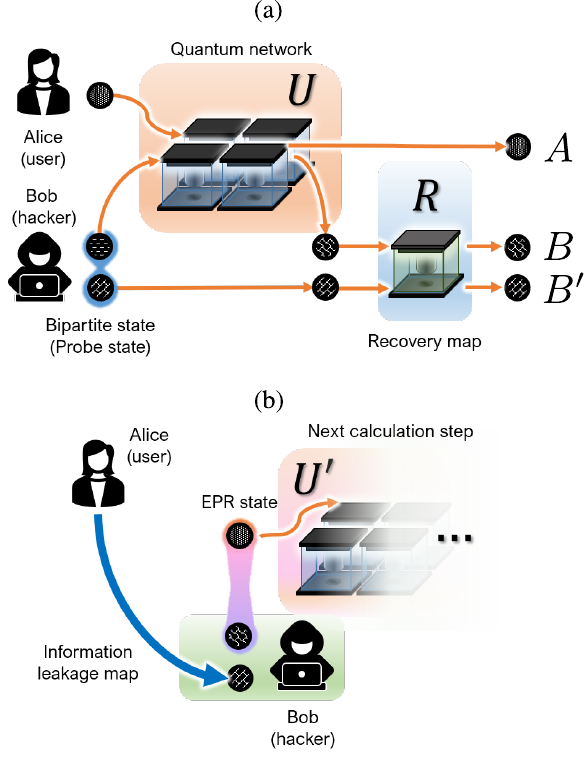}
	\caption{(a) Schematic diagram of quantum hacking of a unitary process $U$ (quantum network or quantum computer) in the two-user scenario. (b) Ideally, Bob, the hacker would successfully extract Alice's information and plant a maximally entangled state for the next quantum computation.}
	\label{fig:QHscheme}
\end{figure}

The next optimal strategy for Bob is to build as much correlation as possible with the target system $A$ while extracting quantum information out of it (See Fig.~\ref{fig:QHscheme}). This is because building correlations allows the extraction of quantum information from the next (undecided) computation step provided Bob has access to a part of the current computation output. In general, Bob possesses a reference system $B'$ ($d_B=d_{B'}$) and prepares an entangled input probe state $\ket{\phi}_{BB'}$. Bob's goals are, therefore, to extract the \emph{input} information stored in $A$, and install an \emph{output} maximally entangled state in $AB$. The latter can be interpreted as an extraction of quantum data of \textit{future} quantum computation on $A$~\cite{horodecki2005partial}, because having a maximally entangled state with the target system yields the maximum side information.

To achieve both goals, Bob applies a unitary recovery map $R$ on systems $BB'$. Since the input data in $A$ is unknown, for the purpose of evaluating Bob's strategy following the standard approach \cite{hayden2007black, yoshida2017efficient}, we assume that it is in a maximally entangled $\ket{\psi}_{AA'}$ state with an environment $A'$ ($d_A=d_{A'}$), where $\ket{\psi}_{XY}=\sum_{i=1}^{d'}\ket{ii}_{XY}/\sqrt{d'}$ 
and $d'=\min\{d_X,d_Y\}$ for any systems $XY$. For this case, successful data extraction from $A$ means entanglement swapping; Bob should have $\ket{\psi}_{A'B'}$ in the final stage. The fidelity with maximally entangled inputs is known to be a monotone function of the average fidelity of a pure-state input \cite{horodecki1999general}. Consequently, we define the quantum-hacking fidelity as
\begin{equation} 
    \label{eqn:phack}
    p_{\mathrm{hack}}:=|\bra{\psi}_{AB}\bra{\psi}_{A'B'}R_{BB'}U_{AB}\ket{\psi}_{AA'}\ket{\phi}_{BB'}|^2.
\end{equation}

Since the fidelity never decreases under a partial trace, $p_\mathrm{hack}$ serves as a lower bound for both fidelities of the extracted quantum data (systems $A'B'$) and implemented entangled state (systems $AB$). We can parametrize any bipartite entangled pure state $\ket{\phi}_{BB'}$ with an operator $\chi$ acting on system $B'$ such that $\ket{\phi}_{BB'}=\sum_i \ket{i}_B\otimes \chi\!\ket{i}_{B'}$  with $\|\chi\|_2=1$. Here $\|X\|_p:=(\Tr|X|^p)^{1/p}$, where $|X|:=\sqrt{X^\dag X}$, is the Schatten $p$-norm. With these, Eq.~(\ref{eqn:phack}) is simplified to
\begin{equation} \label{eqn:PVhack}
    p^{(R,\chi)}_{\mathrm{hack}}=|\Tr[R(I_B\otimes \chi)U^o]|^2/d_A^3.
\end{equation}
Here, the (generally non-unitary) map $U^o : AA' \to BB'$ is represented by a matrix, understood as a tensor, formed by cyclically rotating the indices of $U$ clockwise by one position---${U^o}^{ij}_{kl}:=U^{ki}_{lj}$. Here, $X^{ij}_{kl}:=\bra{ij}X\ket{kl}$ in the computational basis~\cite{SM_Qhack}. This amounts to rotating $U$ clockwise by 90~degrees in tensor-network diagrams~\cite{montangero2018introduction, landsberg2011geometry}, and is closely related to tensor reshuffling~\cite{zyczkowski2004duality, miszczak2011singular, bruzda2009random}. Note that $\|U^o\|_2=\|U\|_2=\sqrt{d_Ad_B}$~\bibnote{Although $R$ is $d_B^2\times d_B^2$, it only acts on a $d_A^2$-dimensional subspace $\Im \{(I_B \otimes \chi)U^o\}$ to the right and $(\mathrm{Ker } \{U^o\})^\perp$ to the left, so that we may treat $R$ either as a rank-$d_A^2$ partial unitary matrix or a $d_A^2 \times d_B^2$ coisometry without losing generality.}.

Each pair $R,\chi$ constitutes a hacking strategy for Bob. For a given $\chi$, which is identical to fixing Bob's probe state, the optimal unitary recovery~$R$ is the one that gives the polar decomposition $(I_B\otimes \chi)U^o=R^\dag|(I_B\otimes \chi)U^o|$. This leads to
\begin{equation}
    p^{(\chi)}_{\mathrm{hack}}=\max_R\, p^{(R,\chi)}_{\mathrm{hack}}=\|(I_B\otimes \chi)U^o\|_1^2/d_A^3\,,
    \label{eqn:phack_optR}
\end{equation}
which is equivalent to inverting a possibly non-unitary operator with a unitary one~\cite{lesovik2019arrow}. This leaves the problem of finding an optimal $\chi$ that achieves the largest hacking fidelity. A natural candidate would be $\chi=I_B/\sqrt{d_B}$, which corresponds to a maximally entangled probe state. Equation~(\ref{eqn:phack_optR}) then immediately yields the fidelity $p_{\mathrm{hack}}^\textsc{me}=\|U^o\|_1^2/(d_A^{3}d_B)$, which is a \emph{unitarity measure} of $U^o$ as $\|U^o\|_1$ is the maximal inner product of $U^o$ and an isometry. With this strategy, perfect hacking ($p_{\mathrm{hack}}^\textsc{me}=1$) is only possible when $U^o$ is proportional to an isometry---$U^{o\dag} U^o=(d_B/d_A)I_{AA'}$. On the contrary, $p_{\mathrm{hack}}^\textsc{me}$ reaches its minimum $1/d_A^2$ when $U^o$ is rank-1, which happens if $U=I_A\otimes I_B$, for instance. This value serves as a lower bound of the optimal hacking fidelity. Physical intuition may lead to the putatively obvious conclusion that a maximally entangled probe state is optimal for quantum hacking. As it turns out, this is, however, not the case in general. For example, for a qubit-qudit controlled unitary operator given as $U_c=I_A\otimes\dyad{0}_B+X_A\otimes(I_B-\dyad{0}_B)$, with the Pauli $X$ operator acting on $A$, $p_{\mathrm{hack}}^\textsc{me}$ is smaller than the $p^{(\chi)}_{\mathrm{hack}}$ with $\chi=(\dyad{0}_B+\dyad{1}_B)/\sqrt{2}$. 

To maximize $p^{(R,\chi)}_\mathrm{hack}$ in Eq.~(\ref{eqn:PVhack}), recalling that $\|\chi\|_2=1$, we may invoke the Cauchy--Schwarz inequality, $|\Tr_{B'}[\chi \Tr_B[U^oR]]|\leq \|\Tr_B[U^oR]\|_2$. This bound is saturated when $\chi=\Tr_B[R^\dag U^{o\dag}]/\|\Tr_B[U^oR]\|_2$. Hence, the true optimal hacking fidelity reads
\begin{equation} \label{eqn:opthack}
    p_{\mathrm{hack}}^{\mathrm{opt}}=\max_R\|\Tr_B[U^o R]\|_2^2/d_A^3,
\end{equation}
where the maximization is over all $d_A^2 \times d_B^2$ coisometry operator $R$~($RR^\dag=I_{AA'}$). By exploiting the polar decomposition once more, a canonical choice of $R$ is $U^oR=|U^{o\dag}|$ and yields the fidelity $p^{\mathrm{PG}}_{\mathrm{hack}}={\|\Tr_B|U^{o\dag}|\|_2^2}/{d_A^3}$. As we shall soon demonstrate that this hacking strategy is near-optimal, we will call this the ``pretty good'' (PG) strategy. As this strategy also outperforms that using a maximally entangled probe state, the following inequalities hold:
\begin{equation} \label{eqn:bds}
    p^{\textsc{ME}}_{\mathrm{hack}}\leq p^{\mathrm{PG}}_{\mathrm{hack}} \leq p^\mathrm{opt}_{\mathrm{hack}}\,.
\end{equation}
Note that the PG strategy \emph{is optimal} when $U^o$ is proportional to an isometry $(p^{\mathrm{ME}}_{\mathrm{hack}}=p_\mathrm{hack}^\mathrm{PG}=p_\mathrm{hack}^\mathrm{opt}=1)$. If $d_A=d_B$, the inequality ${1-p_\mathrm{hack}^\textsc{ME}}\leq 4({1-p_\mathrm{hack}^\mathrm{opt}})$ implies that near-perfect hacking $(p_\mathrm{hack}^\mathrm{opt} \approx 1)$ is possible only when $U^o$ is nearly unitary $(p_\mathrm{hack}^\textsc{ME} \approx 1)$~\cite{SM_Qhack}.

\begin{figure}[t]
	\centering\includegraphics[width=0.85\columnwidth]{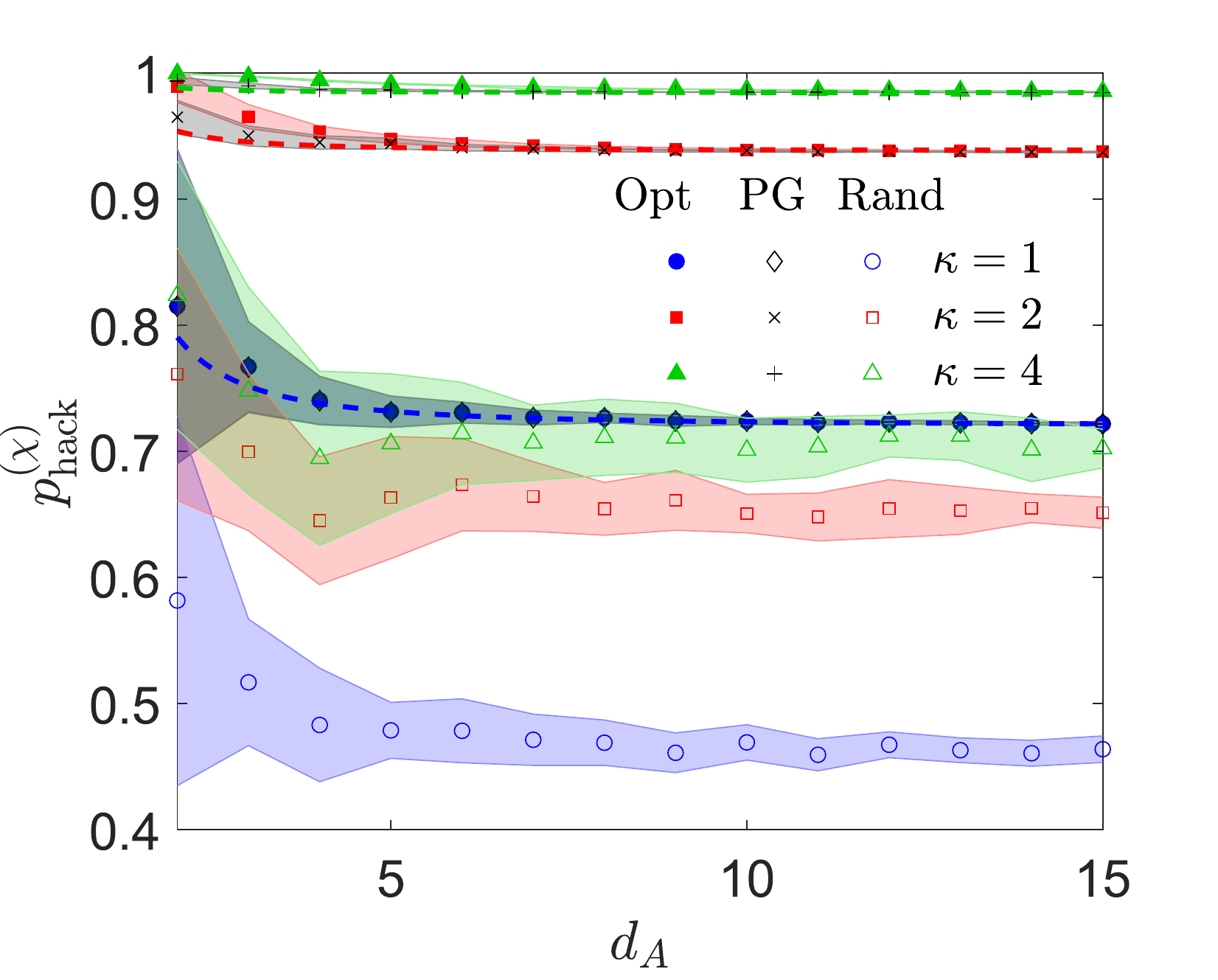}
	\caption{\label{fig:avg_phackopt}Averaged quantum-hacking performance (over 20 randomly generated Haar unitary networks $U$'s) featuring the optimal strategy~(Opt) \emph{via}~\eqref{eqn:opthack}, the PG strategy with $\widetilde{\chi}$, and a random one~(Rand) using an arbitrarily-chosen probe state. When $\kappa=1$, PG is almost the same as Opt in hacking performance. As $\kappa$ increases,  $p^\mathrm{opt}_\mathrm{hack}\rightarrow\mathcal{I}_\kappa^2\rightarrow1$. All theoretical dashed curves are computed with~\eqref{eqn:asymp}.}
\end{figure}

{\it Optimal hacking performance}.---For a quantum network described by a generic unitary $U$ and an optimal hacking strategy defined by optimizing $R$ and $\chi$, the hacking fidelity $p_\mathrm{hack}$ defined in Eq.~\eqref{eqn:phack} reaches an optimal value $p^\mathrm{opt}_\mathrm{hack}$ stated in Eq.~\eqref{eqn:opthack}. As previously discussed, the optimal probe state~($\chi_\mathrm{opt}$) for $p^\mathrm{opt}_\mathrm{hack}$ is generally not maximally entangled ($\chi_\mathrm{opt}\neq I_B/\sqrt{d_B}$). While the general solution of $\chi_\mathrm{opt}$ for the maximization problem in Eq.~\eqref{eqn:opthack} has no known analytical form, we derive an iterative gradient-ascent algorithm~\cite{teo2011mlme,teo2011adaptive} to acquire $\chi_\mathrm{opt}$ in~\cite{SM_Qhack}. 

Interestingly, one can calculate an analytical form of $p^\mathrm{opt}_\mathrm{hack}$ for sufficiently large dimensions ($d_A,d_B\gg1$). To do this, we observe that a maximally entangled probe state is asymptotically optimal for quantum hacking---$\chi_\mathrm{opt}\rightarrow I_B/\sqrt{d_B}$. This follows from the fact that the reduced state of any high-dimensional pure state approaches the maximally mixed state~\cite{page1993entropy}, which then implies that $p^\mathrm{opt}_\mathrm{hack}\rightarrow\| U^o\|_1^2/(d_A^3d_B)$. We shall consider $U$ as a random unitary operator distributed according to the Haar measure of the unitary group. Using properties of this measure and results from random matrix theory~\cite{mehta2004random,marcenko1967eigenvalues,SpFuncBk}, in the two-user scenario where only Alice and Bob are connected to the quantum network, we have the asymptotic Haar-averaged formula for $d_B\geq d_A$,
\begin{equation}
	\overline{p^{\mathrm{opt}}_\mathrm{hack}}\approx\mathcal{I}_\kappa^2+(1-\mathcal{I}_\kappa^2)/d_A^2\,,\,\,\mathcal{I}_\kappa={}_2\mathrm{F}_1\left(1/2,-1/2;2;1/\kappa^{2}\right)\,,
	\label{eqn:asymp}
\end{equation}
with $\kappa=d_B/d_A$ and ${}_2\mathrm{F}_1(\cdot)$ is the hypergeometric function~\cite{SM_Qhack}. Specifically, $\mathcal{I}_1=8/(3\pi)$. If $d_B<d_A$, we instead have $\overline{p^{\mathrm{opt}}_\mathrm{hack}}\approx\kappa^2\mathcal{I}_\kappa^2+(1-\mathcal{I}_\kappa^2)/d_A^2$, which appeals to common sense that a hacking space smaller than Alice's user space is inadequate for hacking. 

\begin{figure}[t]
	\centering\includegraphics[width=0.85\columnwidth]{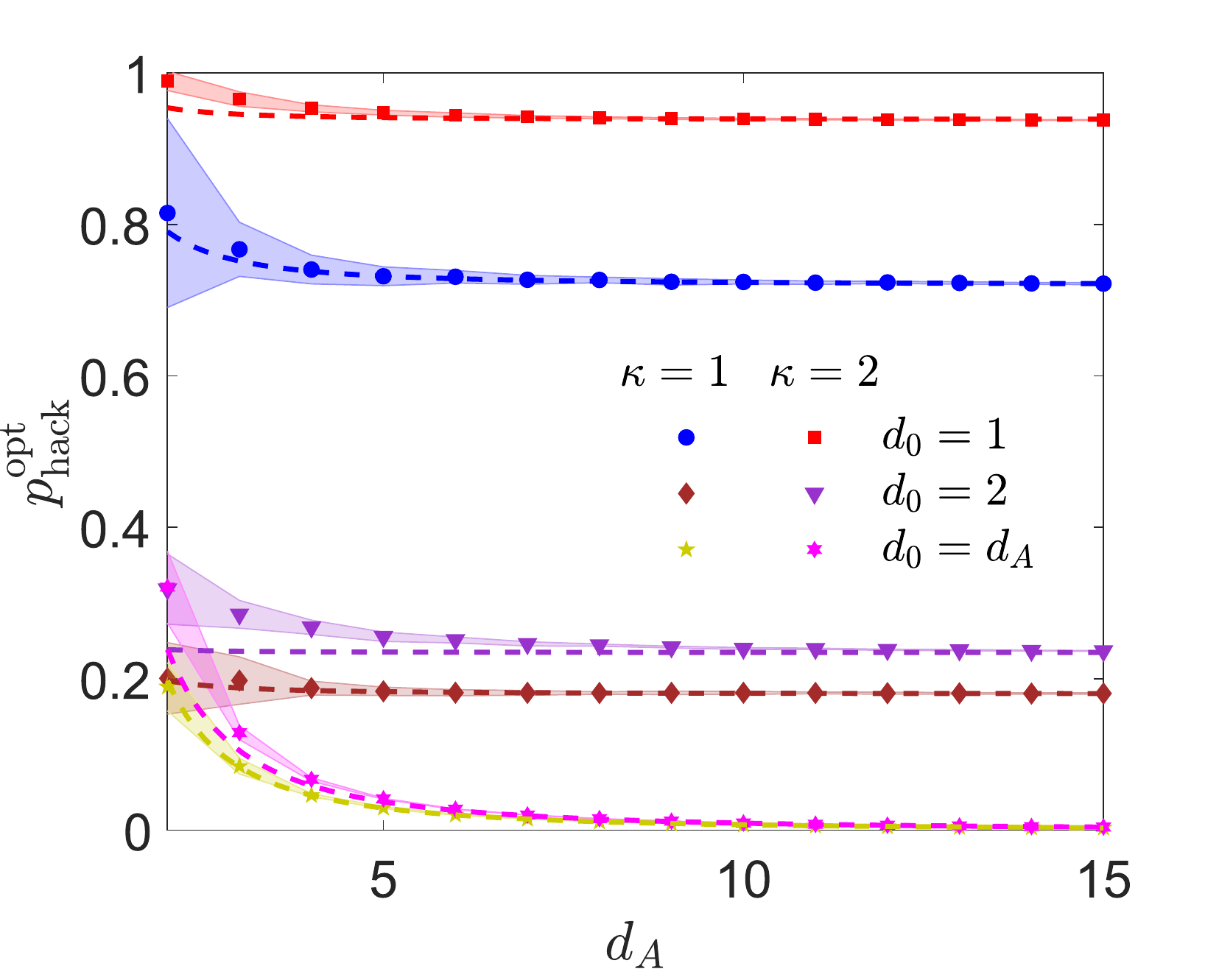}
	\caption{\label{fig:avg_phackd0} Averaged \emph{optimal} quantum-hacking performance for situations where Bob the hacker has restricted amount of network usage ($d_0\geq1$). The case $d_0=1$ corresponds to Charlie's absence. Bob's optimal hacking fidelity drops rapidly with $d_A$ as soon as there exists a third-party user with a Hilbert-space dimension identical to Alice's. All theoretical dashed curves are given by $\overline{p^{\mathrm{opt}}_\mathrm{hack}}(d_0)$.}
\end{figure}

Figure~\ref{fig:avg_phackopt} shows the performances of three hacking strategies with optimal recovery $R$ [see \eqref{eqn:phack_optR}]. The results indicate that efforts in using optimal probe states for a given network $U$ do pay off with a much higher hacking fidelity compared to all other random choices of Bob's probe state. The quantity $\mathcal{I}_\kappa^2$ is an important indicator of the limiting performance for hacking large quantum networks of a fixed dimension ratio $\kappa$. It also suggests that in the two-user scenario, a larger Hilbert space of Bob relative to Alice's results in a larger $p^{\mathrm{opt}}_\mathrm{hack}$. A single-qubit ancilla ($\kappa=2$) is enough to boost $p^{\mathrm{opt}}_\mathrm{hack}$ all the way to $\approx0.936$. We also find that the PG strategy $\widetilde{\chi}:=\chi=\Tr_B|U^{o\dag}|/\|\Tr_B|U^{o\dag}|\|_2$ is almost optimal for any $d_A$ and $d_B$. In particular, when $d_A=2=d_B$, we precisely get $\widetilde{\chi}=\chi_\mathrm{opt}=I_B/\sqrt{d_B}$~\cite{SM_Qhack}.

The hacking landscape becomes very different when there are multiple users connected to the quantum computer. Let us suppose that Charlie, a third party who represents either a single user or user group, is using the computer described by some unitary $U$, and Bob is interested in hacking only Alice's data. In this case, Charlie's data, encoded in a $(d_C=d_0)$-dimensional space, should remain untempered after passing through $U$. The resulting hacking fidelity has essentially the same form as that in \eqref{eqn:phack}, only that $U^o$ is now a rotated $\Tr_C U/d_0$. Correspondingly, $\overline{p^{\mathrm{opt}}_\mathrm{hack}(d_0)}=\overline{p^{\mathrm{opt}}_\mathrm{hack}}/d_0^2$~\cite{SM_Qhack}. Figure~\ref{fig:avg_phackd0} illustrates the hacking performance when Charlie is present. We note that for $d_0>1$, the asymptotic optimal hacking fidelity is lower than that when $d_0=1$ (no Charlie), which is obvious from the $1/d_0^2$ dependence in $\overline{p^{\mathrm{opt}}_\mathrm{hack}(d_0)}$. If Charlie uses a Hilbert space of the same dimension as Alice's to encode his data ($d_A=d_0$), we see that $p^{\mathrm{opt}}_\mathrm{hack}=O(1/d_A^2)\rightarrow0$. Such a characteristic becomes conspicuous for an arbitrarily large quantum computer ($d_0\gg1$), and reveals the exceedingly-low plausibility of targeted quantum hacking for large networks.

{\it Duality with Hayden--Preskill protocols.---}Surprisingly, the seemingly harder problem of finding an optimal quantum-hacking strategy by Bob on Alice is equivalent to  that of a Hayden--Preskill-type protocol of Alice on Bob. In this setting we assume that, instead of $U$, its (computational-basis) transpose $(U^\top)^{ij}_{kl}=U^{kl}_{ij}$ is applied to systems $AB$. Now, Alice wants to extract information from Bob's system $B$. Similar to quantum hacking, to model such information extraction, we assume that a maximally entangled state $\ket{\psi}_{BB'}$ is fed into $U^\top$. Alice also chooses a maximally entangled state $\ket{\psi}_{AA'}$ as a probe state. Systems $AB$ interact with $U^\top$ and Alice applies a $d_A^2$-dimensional unitary operator $W$ on $AA'$. Alice's goal is to prepare a maximally entangled state on systems $AB'$. The optimal fidelity between the actual and ideal states is $p_{\mathrm{HP}}^\mathrm{opt}=\max_W\bra{\psi}_{AB'}\Tr_{A'B}[\mathcal{W}\circ\mathcal{U}(\dyad{\psi}_{A'ABB'}^{\otimes 2})]\ket{\psi}_{AB'}$.
Here, $\mathcal{W}(\rho):=W_{AA'}\rho\, W_{AA'}^\dag\equiv W\rho\,W^\dag$ and $\mathcal{U}(\rho):= U^\top_{AB}\rho\,U_{AB}^*\equiv U^\top\rho\,U^*$. This expression can also be simplified in terms of the \texttt{SWAP} operator $F$ to
\begin{equation} \label{HPopt}
    p_{\mathrm{HP}}^\mathrm{opt}=\max_W\|\Tr_B[U^o\,W^\top\! F]\|_2^2/(d_A d_B^2).
\end{equation}
It follows that the optimal hacking strategy of Eq.~(\ref{eqn:opthack}) and the optimal strategy to $p_{\mathrm{HP}}$ are related by $R=W^\top\!F$. Therefore, quantum hacking is a dual problem of the Hayden--Preskill protocol, in the sense that if one problem is solvable for an arbitrary $U$, then so is the other. It follows that $p_{\mathrm{HP}}^\mathrm{opt}=\opt/\kappa^{2}$ and $p_{\mathrm{HP}}^\mathrm{opt}<1$ when $d_B>d_A$.

{\it Quantum hacking and entanglement recycling}.---The operator $U^o$ appears in various scenarios. In some, the input and output systems of the unitary operator $U$ need not match. Let $U:AB\to KL$ be a map of dimension $D=d_Ad_B=d_Kd_L$. By using a probe state ($\chi$) and recovery map $R$ on $LB'$ with matching dimensions, we get the modified hacking fidelity $p_\mathrm{hack}^{(R,\chi)}={|\Tr[R(I_L\otimes \chi)U^o]|^2}/{(d_A^2 d_K)}$. One remarkable example is the evaporation of an old black hole, where the probe state is  maximally entangled between the inner degrees of freedom of the black hole and all Hawking radiation emitted from the back hole up to that point. The black-hole degrees of freedom is typically much larger than those of matter falling into it momentarily. Let the former be $D_B=d_B=d_K$ and the latter be, say, qudit: $d_M=d_A=d_L$. If Bob collects an additional $d_M$-dimensional Hawking radiation, the resulting optimal hacking fidelity is
$p_\mathrm{hack}^\mathrm{BH}=\|U^o\|_1^2/D^2$.

We remark that the dimension of the black-hole interior state remains the same since a qudit enters the black hole and another exits it. Depending on the assumptions made on the dynamics of black holes (see \cite{tajima2021symmetry} for the recent discussion on the effect of symmetry for information recovery), there may be an estimated value $1-\epsilon_{\mathrm{BH}}$ of $p_\mathrm{hack}^\mathrm{BH}$. For example, for Haar random $U$, $p_\text{hack}^\text{BH}$ tends to $(8/3\pi)^2\approx 0.72$ for large $D$~\cite{SM_Qhack}. This also serves as a lower bound for the fidelity between the posterior probe state and a maximally entangled state. If one uses a probe state whose maximal fidelity with a maximally entangled state is $f$ for the information extraction of the next qudit falling into the black hole, where optimal hacking fidelity is typically $p_\mathrm{hack}^\mathrm{BH}$ for large $D$, then the optimal hacking fidelity approaches to the product $fp_\mathrm{hack}^\mathrm{BH}$. However, for a given hacking fidelity $p_\mathrm{hack}$ and the fidelities of the extracted quantum data ($f_\mathrm{ext}$), and between the posterior probe state and a maximally entangled state~($f_\mathrm{post}$), the following trade-off relation exists~\cite{SM_Qhack}:
\begin{equation}\label{eqn:fidtrade}
    \sqrt{1-f_\mathrm{ext}} + \sqrt{1-f_\mathrm{post}} \geq 2(1-p_\mathrm{hack})/3.
\end{equation}
So one should choose between accurate data extraction and good entanglement recycling for imperfect hacking; giving up the former means inaccurate data extraction for the current round of hacking, and giving up the latter leads to a worse fidelity in the next round.

\textit{Discussion.---}We proposed a quantum hacking task, which entails the extraction and replacement of quantum data through limited interaction with a quantum network, and analyzed its performance in terms of the hacking fidelity. In finding good hacking strategies and calculating the optimal hacking fidelity for generic multipartite unitary networks, we give explicit operational meaning to 90-degree tensor rotations of unitary operators. While Bob's hacking fidelity on Alice saturates at a nonzero value in the two-network-user scenario, we find that with multiple users, naive attempts to hack any single user would generally result in exceedingly-low hacking fidelity. To improve the hacking success, it is necessary to perform program modifications to $U$, akin to hackers introducing malware to control classical computers. For quantum networks, quantum circuits would constitute such a program, but since an arbitrary quantum program cannot be encoded into a state owing to the no-programming theorem~\cite{nielsen1997program}, Bob would need to supplement his quantum resources with additional classical attacks to improve the hacking fidelity.

As an interesting application, we considered an information-reflection model for black holes, and surveyed the sustainability of black-hole mirroring. From our trade-off relation in~(\ref{eqn:fidtrade}) and analysis of black-hole hacking, we conclude that the black hole in this model indeed functions as a mirror \cite{hayden2007black}, but its ``reflectivity'' may be gradually degraded over time (See Ref. \cite{oshita2020reflectivity, wang2020echoes} for different notions of reflectivity of quantum black holes). One could collect more Hawking radiation to increase the hacking fidelity, but that necessitates an entanglement reduction~\cite{hayden2007black}, thereby leading to yet another type of reflectivity degradation. This questions whether quantum information gets destroyed when it falls into an `older' black hole that has already reflected a significant amount of quantum information that fell into it. 

\begin{acknowledgments}
This work is supported by the National Research Foundation of Korea (NRF-2019R1A6A1A10073437, NRF-2019M3E4A1080074, NRF-2020R1A2C1008609, NRF-2020K2A9A1A06102946) via the Institute of Applied Physics at Seoul National University and by the Ministry of Science and ICT, Korea, under the ITRC (Information Technology Research Center) support program (IITP-2020-0-01606) supervised by the IITP (Institute of Information \& Communications Technology Planning \& Evaluation).
\end{acknowledgments}

\newpage
\begin{center}
	\bf SUPPLEMENTAL MATERIAL
\end{center}

\section{Implausibility of arbitrary quantum-state installation}
For any $d$-dimensional quantum channel $\Lambda$, the infidelity $1-\mathcal{F}_{\mathrm{ent}}(\Lambda)=1-\bra{\psi}(\mathcal{I}\otimes\Lambda)(\dyad{\psi}))\ket{\psi}$ for maximally entangled input state $\ket{\psi}=\sum_{i=1}^d\ket{ii}$ and the average infidelity over Haar random pure input state $1-\mathcal{F}_{\mathrm{pure}}(\Lambda)=1-\int d\phi \bra{\phi}\Lambda(\dyad{\phi})\ket{\phi}$ have the following linear dependence \cite{horodecki1999general},
\begin{equation}1-F_\mathrm{ent}(\Lambda)=\frac{d}{d+1}(1-\mathcal{F}_\mathrm{pure}(\Lambda))\,,
\end{equation}
we will use the infidelity $1-F_\mathrm{ent}(\Lambda)$ instead and have the equivalent result without losing generality.

Suppose, with the unitary operators $U$, $P$ and $R$, that the \texttt{SWAP} operator can be approximated with error $\epsilon$ according to

\begin{equation}
	\begin{quantikz}[column sep=0.3cm]
		\lstick{$A$} & \qw   &\gate[2]{U}&\qw&\qw          \qw\\
		\push{\ket{0}_B} & \gate[2]{P}   &    &\gate[2]{R}&\meterD{\Tr}\\
		\lstick{$C$} &   &\qw    & &\qw                    \qw
	\end{quantikz} \approx_\epsilon \begin{quantikz} [column sep=0.3cm]
		\lstick{A}&\gate[swap, style={draw=white}]{}&\qw\\
		\lstick{C}&&\qw
	\end{quantikz}.
\end{equation}
Here, $\approx_\epsilon$ means that the two circuits are close to each other with error $\epsilon$ in the infidelity for maximally entangled input states. It means that $\Sigma \approx_\epsilon \Lambda$ is equivalent to $F(J_\Sigma,J_\Lambda)\geq 1-\epsilon$ where $J_\mathcal{N}$ is the normalized Choi matrix for quantum channel $\mathcal{N}$. From the Uhlmann theorem \cite{uhlmann1976transition}, it follows that there exists a pure state $\ket{s}_B$ such that
\begin{equation}
	\begin{quantikz}[column sep=0.3cm]
		\lstick{$A$} & \qw   &\gate[2]{U}&\qw&\qw          \qw\\
		\push{\ket{0}_B} & \gate[2]{P}   &    &\gate[2]{R}&\qw\\
		\lstick{$C$} &   &\qw    & &\qw                    \qw
	\end{quantikz} \approx_{\epsilon} \begin{quantikz} [column sep=0.3cm]
		\lstick{$A$}&\gate[swap, style={draw=white}]{}&\qw\\
		\lstick{$C$}&&\qw\\
		\lstick{$\ket{s}_B$}&\qw&\qw
	\end{quantikz}.
\end{equation}
Since the fidelity never decreases under partial trace, it follows that
\begin{equation}
	\begin{quantikz}[column sep=0.3cm]
		\lstick{$A$}&\qw&\gate[2]{U}&\qw\\
		\lstick{$C$}&\gate{\Phi}&&\qw
	\end{quantikz} \approx_{\epsilon} { \begin{quantikz} [row sep=0.4cm, column sep=0.3cm]
			\lstick{$A$}&\gate[swap, style={draw=white}]{}&\qw&\qw\\
			\lstick{$C$}&&\gate{\Psi}&\qw
	\end{quantikz}}
\end{equation}
where
\begin{equation}
	\begin{quantikz}[column sep=0.3cm] \lstick{$C$}&\gate{\Phi}&\qw\rstick{$B$}\end{quantikz}=
	\begin{quantikz}[column sep=0.3cm]
		\lstick{$\ket{0}_B$}&\gate[2]{P}&\qw\\
		\lstick{C}&&\meterD{\Tr}
	\end{quantikz},
\end{equation}
and
\begin{equation}
	\begin{quantikz}[column sep=0.3cm] \lstick{$C$}&\gate{\Psi}&\qw\end{quantikz}=
	\begin{quantikz}[column sep=0.3cm]
		\lstick{$\ket{s}_B$}&\gate[2]{R^\dag}&\meterD{\Tr}\\
		\lstick{$C$}&&\qw
	\end{quantikz}.
\end{equation}
From the cyclicity of the fidelity for the maximally entangled input state, it follows that
\begin{equation}
	\begin{quantikz}[column sep=0.3cm]
		\lstick{$A$}&\gate[2]{U}&\qw\\
		\lstick{$B$}&&\qw
	\end{quantikz} \approx_{\epsilon} \begin{quantikz} [column sep=0.3cm]
		\lstick{A}&\gate[swap, style={draw=white}]{}&\gate{\Phi^\dag}&\qw\\
		\lstick{B}&&\gate{\Psi}&\qw
	\end{quantikz}.
\end{equation}

It can be interpreted that unless the target network $U$ itself is already close to a swapping operator followed by local operations, it is impossible to nearly perfectly substitute the quantum information out of the network. Conversely, if $U$ is close to the \texttt{SWAP} operator with error $\epsilon$ (if the dimensions of $A$ and $B$ do not match, then it can be $\texttt{SWAP} \oplus I$), then by choosing $P$ and $R$ also as \texttt{SWAP} operators, one can achieve data substitution with error $\epsilon$.

\section{Rotation of matrix and fidelitiy bounds}
We first show how the fidelity expression
\begin{equation}
	p_{\mathrm{hack}}^{(R,\chi)}:=|\bra{\psi}_{AB}\bra{\psi}_{A'B'}R_{BB'}U_{AB}\ket{\psi}_{AA'}\ket{\phi}_{BB'}|^2\;,
\end{equation}
is simplified with the rotated matrix $U^o$. First, note that $p_{\mathrm{hack}}$ is the fidelity between the pure states
\begin{equation}
	\begin{quantikz}
		&\makeebit[0]{$\dfrac{1}{\sqrt{d_A}}$}&\qw&\qw&\qw\rstick{$A'$} \\
		&&\gate[2]{U}&\qw&\qw\rstick{$A$}\\
		&\makeebit[0]{}&\qw&\gate[2]{R}&\qw\rstick{$B$}\\
		&&\gate{\chi}&&\qw\rstick{$B'$}
	\end{quantikz}\text{ and }
	\begin{quantikz}[row sep=0.75cm]
		&\makeebit[0]{}&\qw&\qw\rstick{$A'$} \\
		&\makeebit[0]{$\dfrac{1}{d_A}$}&\makeebit[0]{}&\qw\rstick{$A$}\\
		&\makeebit[0]{}&&\qw\rstick{$B$}\\
		&&\qw&\qw\rstick{$B'$}
	\end{quantikz}.
\end{equation}
Here, $\!\!\!\!\!\!\!\!\begin{quantikz}[row sep=0.3cm]&\makeebit[0]{}&\qw\\&&\qw\end{quantikz}=\sum_i \ket{ii}$ represents an unnormalized maximally entangled state with the appropriate Schmidt number for the system it is defined on. Therefore, it can be expressed with a tensor network diagram.
\begin{equation} \label{eqn:phackdiag}
	p_{\mathrm{hack}}^{(R,\chi)}=\frac{1}{d_A^{3}}\!\!\!\!\!
	\begin{quantikz}
		&\makeebit[0]{}&\qw&\qw&\qw&\qw&\qw\makeebit[0]{} \\
		&&\gate[2]{U}&\qw&\qw\makeebit[0]{}&&\makeebit[0]{}\\
		&\makeebit[0]{}&\qw&\gate[2]{R}&\qw&&\makeebit[0]{}\\
		&&\gate{\chi}&&\qw&\qw&\qw
	\end{quantikz}
	\;,
\end{equation}
with the equivalence of $\ket{\phi}_{BB'}$ and $\chi$:
\begin{equation}
	\ket{\phi}_{BB'}=\begin{quantikz}
		\makeebit[0]{}&\qw&\qw\\
		&\gate{\chi}&\qw
	\end{quantikz}\;.
\end{equation}
We remark that the time flows from left to right in the diagram, as opposed to the matrix multiplication order. The definition of $U^o$ can be expressed in a circuit diagram as follows:
\begin{equation}
	\begin{quantikz}
		&\gate[2]{U^o}&\qw\\
		&&\qw\\
	\end{quantikz}=
	\begin{quantikz}
		&\qw&\qw\makeebit[0]{}\\
		&\gate[2]{U}&\qw\\
		\makeebit[0]{}&&\qw\\
		&\qw&\qw
	\end{quantikz}
	\;.
\end{equation}
By plugging this diagram into Eq. (\ref{eqn:phackdiag}), we get
\begin{equation}
	p_{\mathrm{hack}}^{(R,\chi)}=\frac{1}{d_A^{3}}\!\!\!\!\!
	\begin{quantikz}
		&\makeebit[0]{}&\qw&\qw&\qw&\qw\makeebit[0]{}\\
		&&\gate[2]{U^o}&\qw&\gate[2]{R}&\qw\\
		&\makeebit[0]{}&&\gate{\chi}&&\qw\makeebit[0]{}\\
		&&\qw&\qw&\qw&\qw
	\end{quantikz}\;.
\end{equation}
This is equivalent to the expression
\begin{equation}
	p_\mathrm{hack}^{(R,\chi)}=\frac{|\Tr[R(I_B \otimes \chi)U^o]|^2}{d_A^3}.
\end{equation}

Since $\sqrt{d_B}\|\Tr_B|U^{o\dag}|\|_2\geq\|\Tr_B|U^{o\dag}|\|_1=\|U^{o}\|_1$, $p_\mathrm{hack}^\textsc{PG}=\|\Tr_B|U^{o\dag}|\|_2^2/d_A^3$ is higher than $p^{\mathrm{ME}}_\mathrm{hack}=\|U^o\|_1^2/(d_A^3d_B).$ Also, since the PG strategy is a particular strategy, the fidelity of it is not larger than that of the optimal strategy, so we have $\PG \leq \opt$. In summary, we have

\begin{equation}
	\ME \leq \PG \leq \opt.
\end{equation}

When $d_A=d_B=d,$ $\opt$ is the maximal fidelity between a maximally entangled state with the Schmidt rank $d^2$ and a pure state of the form $\Omega_\chi:=d(U_{AB}\otimes\chi_{B'})\dyad{\psi}^{\otimes 2}_{AA'BB'}(U_{AB}\otimes\chi_{B'})^\dag$ with $\|\chi\|_2=1$. Let $\chi_M$ be a $\chi$ that achieves the maximum.  Since the partial trace never decreases the fidelity, by tracing out systems other than $B'$, we get $F(I_{B'}/d,|\chi_M|^2)\geq \opt$. Let the recovery map that achieves the optimal fidelity be $R_\mathrm{opt}$ and let $\Theta_{V}:= V_{BB'}\dyad{\psi}_{AA'BB'}^{\otimes2} V_{BB'}^\dag$ for any bipartite unitary operator $V_{BB'}$, so that $\opt = F(\Omega_{\chi_M},\Theta_{R_\mathrm{opt}})$. Since there is a freedom of local unitary operation to the choice of $R_\mathrm{opt}$ and $\chi_M$, without loss of generality, we can assume that $\chi_M$ is positive semi-definite so that $\Tr \chi_M = \Tr |\chi_M|.$

From the following relation for arbitrary pure quantum states $\ket{\eta_1}$ and $\ket{\eta_2}$,
\begin{equation} \label{eqn:puredist}
	\frac{1}{2}\|\dyad{\eta_1}-\dyad{\eta_2}\|_1 = \sqrt{1-|\bra{\eta_1}\ket{\eta_2}|^2},
\end{equation}
we have $\sqrt{1-\ME}=\min_W \|\Theta_W-\Omega_{\textsc{ME}}\|_1/2$, where $\Omega_\textsc{ME}:=\Omega_{{I_{B'}}/{\sqrt{d}}}$. Therefore $\sqrt{1-\ME}\leq\|\Theta_{R_\mathrm{opt}}-\Omega_\textsc{ME}\|_1/2$. Because of the triangle inequality, we have  $\|\Theta_{R_\mathrm{opt}}-\Omega_\textsc{ME}\|_1 \leq \|\Theta_{R_\mathrm{opt}}-\Omega_{\chi_M}\|_1+\|\Omega_{\chi_M}-\Omega_\textsc{ME}\|_1$. By Eq. (\ref{eqn:puredist}), we have $\|\Theta_{R_\mathrm{opt}}-\Omega_{\chi_M}\|_1/2 = \sqrt{1-\opt}$ and  $\|\Omega_{\chi_M}-\Omega_\textsc{ME}\|_1/2 = \sqrt{1-F(I_{B'}/d,|\chi_M|^2)} \leq \sqrt{1-\opt}$. As a result, we have $\sqrt{1-\ME}\leq 2\sqrt{1-\opt}$ thus $1-\ME \leq 4(1-\opt).$

The trade-off relation between the data extraction fidelity $f_\mathrm{ext}$ and the posterior probe state fidelity $f_\mathrm{prob}$

\begin{equation}
	\sqrt{1-f_\mathrm{ext}} + \sqrt{1-f_\mathrm{prob}} \geq \frac{2}{3}(1-\opt)\;,
\end{equation}
directly follows from the following result.
\begin{theorem}
	For arbitrary bipartite quantum state $\rho_{AB}$ and pure states $\ket{\psi}_A$ and $\ket{\phi}_B$, let $F_A:=\bra{\psi}\rho_A\ket{\psi}$, $F_B:=\bra{\phi}\rho_B\ket{\phi}$ and $F_{AB}:=\bra{\psi}_A\bra{\phi}_B\rho_{AB}\ket{\psi}_A\ket{\phi}_B$. Then the following inequality holds
	\begin{equation}
		\sqrt{1-F_A} + \sqrt{1-F_B} \geq \frac{2}{3}(1-F_{AB}).
	\end{equation}
\end{theorem}
\begin{proof}
	Let $\psi_A:=\dyad{\psi}$ and for other pure states with similar notations. By applying the Uhlmann theorem \cite{uhlmann1976transition}, it follows the existence of a (possibly mixed) state $\phi_B'$ such that $F_A=F(\psi_A \otimes \phi_B', \rho_{AB})$. From the Fuchs-van de Graaf inequality \cite{fuchs1999cryptographic}, we have $\|\psi_A\otimes\phi'_B-\rho_{AB}\|_1/2\leq \sqrt{1-F_A}$. From the monotonicity of 1-norm under partial trace, we have $\|\phi'-\rho_B\|_1/2\leq\sqrt{1-F_A}$. Therefore $\|\phi_B-\phi_B'\|_1/2\leq\|\phi_B-\rho_B\|_1/2+\|\rho_B-\phi_B'\|_1/2\leq\sqrt{1-F_B}+\sqrt{1-F_A}.$ From this we can bound the distance $\|\psi_A\otimes\phi_B-\rho_{AB}\|_1/2\leq \|\psi_A\otimes(\phi_B-\phi_B')\|_1/2+\|\psi_A\otimes\phi_B'-\rho_{AB}\|_1/2.$ Here, $\|\psi_A\otimes(\phi_B-\phi_B')\|_1/2=\|\phi_B-\phi_B'\|_1/2\leq\sqrt{1-F_A}+\sqrt{1-F_B}$ and $\|\psi_A\otimes\phi'_B-\rho_{AB}\|_1/2\leq \sqrt{1-F_A}$, so we have
	\begin{equation}
		\|\psi_A\otimes\phi_B-\rho_{AB}\|_1/2\leq 2\sqrt{1-F_A}+\sqrt{1-F_B}.
	\end{equation}
	Using the same argument we can also derive
	\begin{equation}
		\|\psi_A\otimes\phi_B-\rho_{AB}\|_1/2\leq \sqrt{1-F_A}+2\sqrt{1-F_B},
	\end{equation}
	hence, by averaging two inequalities, and again using the Fuchs--van de Graaf inequality $1-F_{AB}\leq\|\psi_A\otimes\psi_B-\rho_{AB}\|_1/2$, since $\psi_A\otimes\phi_B$ is a pure state, we get the desired result.
\end{proof}

Now, consider the black hole radiation problem of \textit{Hacking as entanglement recycling} section. When the probe state is a general mixed bipartite state $\Pi_{BB'}=\sum_i p_i \dyad{\phi_i}_{BB'}$ with $\ket{\phi_i}_{BB'}=\sum_k(I_B\otimes \chi_i)\ket{kk}_{BB'}$, the hacking fidelity is given as
\begin{equation}
	p_\mathrm{hack}^{(R,\Pi)}=\sum_i p_i \frac{|\Tr[R (I_B \otimes \chi_i) U^o]|^2}{d_M^2D_B}.
\end{equation}
If the dimension of the Hilberst space of black hole state is large enough, then the PG strategy becomes nearly optimal thus, $p_\mathrm{hack}^{(R,\Pi)}$ reduces to $\sum_i p_i |\Tr[\chi_i \Tr_B |U^{o\dag}|]|^2/d_M^3.$ Moreover, as $D\to \infty$. $\Tr_B|U^{o\dag}|$ converges to $\|U^o\|_1 I_B'/D_B$(See Sec. \ref{app:asymp}), so we have
\begin{align}
	\max_R p_\mathrm{hack}^{(R,\Pi)}&\approx\sum_i p_i \frac{|\Tr\chi_i|^2\|U^o\|_1^2}{D_B D^2}\nonumber\\
	&=\ME\sum_i p_i \frac{|\Tr\chi_i|^2}{D_B}\nonumber\\
	&\approx \opt f_\mathrm{prob}'.
\end{align}
Where $f_\mathrm{prob}'=\sum_i p_i {|\Tr\chi_i|^2}/{D_B}$ is the fidelity between $\Pi_{BB'}$ and a maximally entangled state. Therefore the hacking fidelity is asymptotically the product of the optimal hacking fidelity and $f_\mathrm{prob}'$.

\section{Numerical maximization of $p^\mathrm{opt}_\mathrm{hack}$}
\label{app:numer}

Given a quantum computer or network described by $U$, it is possible to derive an iterative numerical scheme to obtain the optimal probe state ($\chi_\mathrm{opt}$) that achieves the optimal hacking fidelity $p^\mathrm{opt}_\mathrm{hack}$. Rather than directly solving the numerical problem in~(5) of the main text, we can instead start with $f_\chi=\|(I_B\otimes\chi)U^o\|_1$, which is the objective function involving the square root of the rightmost side in~(4), and perform a variation with respect to $\chi$. Furthermore, the constraint $\|\chi\|_2=1$ invites the following parametrization $\chi=Z/\|Z\|_2$, such that 
\begin{equation}
	\updelta\chi=\dfrac{\updelta Z}{\|Z\|_2}-\dfrac{Z}{2\|Z\|^3_2}\Tr_{B'}[\updelta ZZ^\dag+Z\updelta Z^\dag]\,.
\end{equation}
Upon denoting $M=(I_B\otimes \chi)U^o$, we consequently have
\begin{align}
	\updelta f_\chi=&\,\dfrac{1}{2}\Tr_{B'}\left[\Tr_{B}[|M^\dag|^{-1}MU^{o\dag}]\dfrac{\updelta Z^\dag}{\|Z\|_2}\right]+\mathrm{c.c.}\nonumber\\
	&\,-\dfrac{1}{2}\Tr|M^\dag|\dfrac{\Tr_{B'}[\updelta ZZ^\dag+Z\updelta Z^\dag]}{\|Z\|^2_2}\,,
\end{align}
which leads to the operator gradient
\begin{equation}
	\dfrac{\updelta f_\chi}{\updelta Z^\dag}=\dfrac{1}{2\|Z\|_2}\left(\Tr_{B}[|M^\dag|^{-1}MU^{o\dag}]-\Tr|M^\dag|\dfrac{Z}{\|Z\|_2}\right)
\end{equation}
with respect to $Z^\dag$. Setting it to zero would then gives the extremal equation
\begin{equation}
	\chi=\dfrac{\Tr_{B}[|M^\dag|^{-1}MU^{o\dag}]}{\|\Tr_{B}[|M^\dag|^{-1}MU^{o\dag}]\|_2}\,,
	\label{eqn:ext_eqn}
\end{equation}
which may alternatively be gotten from reasoning with the Cauchy-Schwarz inequality. As $\|(I_B\otimes\chi)U^o\|_1$ is concave in $\rho_{B'}=\chi^\dag\chi$, one can generally expect a convex solution set of $\rho_{B'}$'s that solve \eqref{eqn:ext_eqn}, all of which give the unique maximal fidelity $p^\mathrm{opt}_\mathrm{hack}$.

In other words, $p^\mathrm{opt}_\mathrm{hack}$ is achieved when a solution $\chi=\chi_\mathrm{opt}$ for Eq.~\eqref{eqn:ext_eqn} is obtained. While there are no known closed-form expressions for this solution, we can nevertheless find explicit analytical forms for certain special cases. The most immediate one happens to be the limiting case $d_B\rightarrow\infty$, whence we have $\chi_\mathrm{opt}\rightarrow I_B/\sqrt{d_B}$, since in this limit, $\Tr_BA\rightarrow\Tr A/d_B$ for any bipartite operator $A$ of systems $BB'$. For finite $d_B$, we may still have an estimate for $\chi_\mathrm{opt}\approx\widetilde{\chi}$. A straightforward way to do this is to simply iterate the extremal equation \eqref{eqn:ext_eqn} once by substituting $I_B/\sqrt{d_B}$ for $\chi$ on the right-hand side. This gives us $\widetilde{\chi}=\Tr_B|U^{o\dag}|/\|\Tr_{B}|U^{o\dag}|\|_2$, which is in practice very close to $\chi_\mathrm{opt}$.

To numerically compute the actual $\chi_\mathrm{opt}$, we adopt the steepest-ascent methodology and require that $\updelta f_\chi=\Tr\left[(\updelta f_\chi/\updelta Z^\dag)\updelta Z^\dag+\updelta Z(\updelta f_\chi/\updelta Z)\right]\geq0$. This amounts to defining the increment $\updelta Z:=\epsilon\,\updelta f_\chi/\updelta Z^\dag$ for some small real $\epsilon>0$ that functions as a fixed iteration step size. This allows us to state the iterative equations
\begin{align}
	Z_{k+1}=&\,\left(1-\dfrac{\epsilon}{2}\dfrac{\Tr|M_k^\dag|}{\|Z_k\|_2}\right)Z_k+\dfrac{\epsilon}{2}\Tr_{B}[|M_k^\dag|^{-1}M_kU^{o\dag}]\,,\nonumber\\
	\chi_{k+1}=&\,\dfrac{Z_{k+1}}{\|Z_{k+1}\|_2}
\end{align}
that can be used to converge $\chi_k$ to $\chi_\mathrm{opt}$ starting with $Z_1=I_B$, where a factor of $\|Z_k\|_2$ has been neglected for a suitably chosen magnitude of $\epsilon$. As $\updelta f_\chi=2\epsilon\Tr[|\updelta f_\chi/\updelta Z|^2]>0$ by construction, convergence is guaranteed as long as $\epsilon$ is sufficiently small. Operationally, one can afford to choose a reasonably large $\epsilon$ to increase the convergence rate.

\section{Optimal quantum hacking of two-qubit quantum networks}

The case where $d_A=2=d_B$ presents the unique situation in which one can confirm, indeed, that $\chi_\mathrm{opt}=I_B/\sqrt{d_B}$. To this end, we proceed to construct the exact expression of $|U^{o\dag}|$. Since $UU^\dag=I$, in terms of the product computational basis $\bra{jk}U\ket{lm}=U^{jk}_{lm}$, the basic relation
\begin{equation}
	\sum^1_{l,m=0}U^{j_1k_1}_{lm}U^{j_2k_2*}_{lm}=\delta_{j_1,j_2}\delta_{k_1,k_2}
	\label{eqn:unitarity}
\end{equation}
shall be immensely useful in the subsequent discussion.

Using Eq.~\eqref{eqn:unitarity}, we first note that the product
\begin{align}
	U^oU^{o\dag}=&\sum_{l,m,m'}\!\Big[\ket{0,m}(U^{00}_{lm}U^{00*}_{lm'}+U^{10}_{lm}U^{10*}_{lm'})\bra{0,m'}\nonumber\\[-1ex]
	&\,\qquad+\ket{1,m}(U^{01}_{lm}U^{01*}_{lm'}+U^{11}_{lm}U^{11*}_{lm'})\bra{1,m'}\nonumber\\
	&\,\qquad+\ket{0,m}(U^{00}_{lm}U^{01*}_{lm'}+U^{10}_{lm}U^{11*}_{lm'})\bra{1,m'}\nonumber\\
	&\,\qquad+\ket{1,m}(U^{01}_{lm}U^{00*}_{lm'}+U^{11}_{lm}U^{10*}_{lm'})\bra{0,m'}\Big]\nonumber\\
	\,\widehat{=}&\,\begin{pmatrix}
		\mathbf{A} & \mathbf{B}^\dag\\
		\mathbf{B} & \mathbf{A}^{-1}\mathrm{Det}\mathbf{A}
	\end{pmatrix}
\end{align}
may be characterized, in the product computational basis, by only two $2\times2$ matrices in a highly specific manner, where $\mathbf{B}$ is traceless. Such a structure is absent in higher dimensions. For convenience, we may rewrite
\begin{equation}
	U^oU^{o\dag}\,\widehat{=}\,\bm{1}+\begin{pmatrix}
		\bm{a\cdot\sigma} & \bm{b}^*\bm{\cdot\sigma}\\
		\bm{b\cdot\sigma} & -\bm{a\cdot\sigma}
	\end{pmatrix}
\end{equation}
in terms of dot products ($\bm{v\cdot w}=\bm{v}^\top\bm{w}$) of the vectorial parameters $\bm{a}$ and $\bm{b}$ with the standard vector of Pauli operators $\bm{\sigma}=(\sigma_x,\sigma_y,\sigma_z)^\top$ to separate the matrix representation of $U^oU^{o\dag}$ into the identity and another $4\times4$ traceless matrix, where $\bm{a}$ is real and $\bm{b}$ complex.

With the identity $(\bm{a\cdot\sigma})(\bm{a}'\bm{\cdot\sigma})=\bm{a\cdot a}'\bm{1}+\I\,\bm{a\times a}'\bm{\cdot\sigma}$, it is a straightforward matter to verify that $|U^{o\dag}|$ has the same matrix-representation structure, where all its parameters satisfy the following conditions:
\begin{align}
	|U^{o\dag}|\,\widehat{=}&\,c'\bm{1}+\begin{pmatrix}
		\bm{a}'\bm{\cdot\sigma} & \bm{b}'^*\bm{\cdot\sigma}\\
		\bm{b}'\bm{\cdot\sigma} & -\bm{a}'\bm{\cdot\sigma}
	\end{pmatrix}\,,\nonumber\\
	1=&\,c'^2+|\bm{a}'|^2+|\RE{\bm{b}'}|^2+|\IM{\bm{b}'}|^2\,,\nonumber\\
	\bm{a}=&\,2\,c'\bm{a}'-2\,\RE{\bm{b}'}\bm{\times}\IM{\bm{b}'}\,,\nonumber\\
	\RE{\bm{b}}=&\,2\,c'\RE{\bm{b}'}-2\,\IM{\bm{b}'}\bm{\times}\bm{a}'\,,\nonumber\\
	\IM{\bm{b}}=&\,2\,c'\IM{\bm{b}'}-2\,\bm{a}'\bm{\times}\RE{\bm{b}'}\,.
\end{align}
Here, $\RE{\cdot}$ and $\IM{\cdot}$ respectively denote the real and imaginary parts of the argument. Evidently, in this fortuitously easy yet general two-qubit scenario, we find that $\Tr_B|U^{o\dag}|=2\,c'I_B$, such that $\widetilde{\chi}=\Tr_B|U^{o\dag}|/\|\Tr_B|U^{o\dag}|\|_2=I_B/\sqrt{d_B}=\chi_\mathrm{opt}$.

\section{Asymptotic formulas for $p^\mathrm{opt}_\mathrm{hack}$} \label{app:asymp}

We first suppose the symmetric problem involving a $d_Ad_B$-dimensional unitary operator $U$. In the asymptotic limit $d_A,d_B\rightarrow\infty$ such that $\kappa=d_B/d_A$, according to the discussions in Sec.~\ref{app:numer}, the corresponding optimal quantum-hacking fidelity takes the form $p^\mathrm{opt}_\mathrm{hack}\rightarrow\|U^o\|^2_1/(d_A^3d_B)$. The analytical form of its average value then necessitates calculating the average term $\overline{\|U^o\|^2_1}$ over all random $U$'s distributed according to the Haar measure. We emphasize that since $U^o$ is represented by a $d_B^2\times d_A^2$ matrix that is obtained from just a sequence of index swapping operations, such a rectangular matrix still retains the statistical properties of a Haar unitary matrix elements, namely $\overline{U^{ojk}_{\hphantom{o}lm}}=0$ and $\overline{|U^{ojk}_{\hphantom{o}lm}|^2}=1/(d_Ad_B)$~\cite{mehta2004random}. If we additionally suppose that $\kappa>1$, then in the asymptotic limit, the \emph{eigenvalues} $(\sigma_j)$ of $\kappa^{-1} U^{o\dag}U^o$ are independently and identically distributed according to the Mar{\v c}enko--Pastur distribution~\cite{marcenko1967eigenvalues}:
\begin{equation}
	\sigma_j\sim\dfrac{1}{2\pi}\dfrac{\sqrt{(\lambda_+-x)(x-\lambda_-)}}{\lambda x}\,,\quad \lambda_{\pm}=(1\pm\sqrt{\lambda})^2\,,
	\label{eqn:MPdist}
\end{equation}
where $\lambda=\kappa^{-2}$. With these,
\begin{align}
	\overline{\|U^o\|^2_1}=&\,\kappa\left(\sum^{d_A^2}_{j=1}\overline{\sigma_j}+\sum_{j\neq k}\overline{\sqrt{\sigma_j}}\,\overline{\sqrt{\sigma_k}}\right)\nonumber\\
	=&\,d_Ad_B+(d_A^3-d_A)d_B\mathcal{I}^2_\kappa\,,
\end{align}
where $\overline{\,\,\cdot\,\,\vphantom{M}}$ now translates to an average with respect to the distribution in~\eqref{eqn:MPdist}. The quantity $\mathcal{I}_\kappa=\overline{x^{1/2}}$ refers to the half-moment of this distribution. For completeness, we evaluate the $m$th moment: 
\begin{align}
	&\,\overline{x^m}=\int^{\lambda_+}_{\lambda_-}\,\dfrac{\D x}{2\pi\lambda}\,x^{m-1}\sqrt{(\lambda_+-x)(x-\lambda_-)}\nonumber\\
	=&\,\dfrac{2}{\pi}(1+\lambda)^{m-1}\int^{1}_{-1}\,\D t\left(1+\dfrac{2\sqrt{\lambda}}{1+\lambda}\,t\right)^{m-1}\sqrt{1-t^2}\nonumber\\
	=&\,\left(1+\dfrac{1}{\kappa^2}\right)^{m-1}{}_2\mathrm{F}_1\left(\dfrac{1-m}{2},1-\frac{m}{2};2;\left(\dfrac{2/\kappa}{1+1/\kappa^2}\right)^2\right)\,.
	\label{eqn:mth}
\end{align}
The variable substitution $x=(\lambda_++\lambda_-)/2+(\lambda_+-\lambda_-)t/2$ has been introduced after the second equality in \eqref{eqn:mth}. We emphasize that the last equality in \eqref{eqn:mth} is valid for \emph{any real} $m$ so long as the previous $t$~integral converges. Upon using the identity~\cite{SpFuncBk}
\begin{equation}
	{}_2\mathrm{F}_1(2a,2a+1-\gamma;\gamma;z)=\dfrac{{}_2\mathrm{F}_1\left(a,a+\dfrac{1}{2};\gamma;\dfrac{4z}{(1+z)^2}\right)}{(1+z)^{2a}}\,,
\end{equation}
we get $\overline{x^m}={}_2\mathrm{F}_1(1-m,-m;2;1/\kappa^2)$. Thereafter, the substitution $m=1/2$ nabs us the final answer $\mathcal{I}_\kappa={}_2\mathrm{F}_1\left(1/2,-1/2;2;1/\kappa^2\right)$, so that $\overline{p^{\mathrm{opt}}_\mathrm{hack}}\approx\mathcal{I}_\kappa^2+(1-\mathcal{I}_\kappa^2)/d_A^2$. Moreover, we may simplify this expression further by considering a moderately large $\kappa$, for which the hypergeometric function has the simple second-order approximation  $\mathcal{I}_\kappa\approx1-1/(8\kappa^2)$. This simplification works amazingly well even for $\kappa=1$---$0.875\approx8/(3\pi)$---such that one might as well use this approximation for any $\kappa$.

Now, if $d_B<d_A$, one can go through a similar line of argument and arrive at $\overline{\|U^o\|^2_1}=d_Ad_B+d_A(d_B^3-d_B)\mathcal{I}^2_\kappa$, in which case, we get $\overline{p^{\mathrm{opt}}_\mathrm{hack}}\approx\kappa^2\mathcal{I}_\kappa^2+(1-\mathcal{I}_\kappa^2)/d_A^2$, which tells us that the asymptotic optimal quantum-hacking fidelity is going to be smaller than that when $d_B>d_A$---by a factor of $\kappa^2$ to be more precise. This makes physical sense since Bob, the quantum hacker, now uses a smaller Hilbert space than Alice.

We are now in the position to discuss the asymmetric problem, where the input and output partitions of $U$ can now be different. That is, $U:AB\to KL$, where systems $AB$ are of dimension $d_Ad_B$, and systems $KL$ of dimension $d_Kd_L$. Since unitarity implies that $d_Ad_B=d_Kd_L$, the usage of the Mar{\v c}enko--Pastur law remains unchanged, with the exception that we are to now consider cases where $d_Ad_K\leq d_Bd_L$ and $d_Ad_K> d_Bd_L$. The first case would again give us the familiar form $\overline{p^{\mathrm{opt}}_\mathrm{hack}}\approx\mathcal{I}_\kappa^2+(1-\mathcal{I}_\kappa^2)/(d_Ad_K)$, which means that the asymptotic quantum-hacking fidelity of $\mathcal{I}_\kappa^2$ can still be achieved. The second case, however, leads to the expression $\overline{p^{\mathrm{opt}}_\mathrm{hack}}\approx[d_Bd_L/(d_Ad_K)]\mathcal{I}_\kappa^2+(1-\mathcal{I}_\kappa^2)/(d_Ad_K)$, which reasonably sets the asymptotic fidelity to $[d_Bd_L/(d_Ad_K)]\mathcal{I}_\kappa^2<\mathcal{I}_\kappa^2$. Obviously, both expressions revert to those in the symmetric problem when $d_A=d_K$ and $d_B=d_L$.

Finally, we look at what happens when a third party, Charlie, is also using the quantum computer alongside Alice and Bob. If Bob is still only concerned with Alice's data, then the hacking fidelity is
\begin{equation} 
	\label{eqn:phack_d0}
	p_{\mathrm{hack}}=|\bra{\Psi'}R_{BB'}U_{ABC}\ket{\Psi}|^2\,,
\end{equation}
where $\ket{\Psi}=\ket{\psi}_{AA'}\ket{\phi}_{BB'}\ket{\psi}_{CC'}$ and $\ket{\Psi'}=\ket{\psi}_{AB}\ket{\psi}_{A'B'}\ket{\psi}_{CC'}$. Descriptively, Charlie's \emph{output} state of $U_{ABC}$, after injecting a maximally entangled input state $\ket{\psi}\bra{\psi}_{CC'}$ with his other ancillary system $C'$, should be unaffected by Bob's action. If the dimensions of $C$ and $C'$ are equal to $d_0$, one immediately finds that
\begin{equation}
	\bra{\psi}_{CC'}U_{ABC}\ket{\psi}_{CC'}=\dfrac{1}{d_0}\Tr_C U
\end{equation}
for any $U_{ABC}=U$. Suppose that $U$ is some $Dd_0$-dimensional unitary operator. In order to calculate the Haar average $\overline{p^\mathrm{opt}_\mathrm{hack}}$, we would then require the understanding of the matrix elements of $\Tr_C U/d_0$:
\begin{equation}
	\dfrac{1}{d_0}\Tr_C U=\dfrac{1}{d_0}\sum^{D-1}_{j,k=0}\ket{j}\sum^{d_0-1}_{l=0}U^{jl}_{kl}\bra{k}\,,
\end{equation}
where $U^{jk}_{lm}$ are the matrix elements of $U$ in the product computational basis, and the dimension of $\ket{j}$ is $D$. It is easy to see that $\overline{\sum_lU^{jl}_{kl}}=0$ and
\begin{align}
	\overline{\sum_lU^{jl}_{kl}\sum_{l'}U^{*j'l'}_{\hphantom{*}k'l'}}=&\,\delta_{j,j'}\delta_{k,k'}\sum_l\overline{|U^{jl}_{kl}|^2}\nonumber\\
	=&\,\dfrac{\delta_{j,j'}\delta_{k,k'}}{Dd_0}\sum_l1=\dfrac{\delta_{j,j'}\delta_{k,k'}}{D}\,.
\end{align}
This shows that the important moment averages of $\Tr_C U$ are independent of $d_0$, such that the only modification of the respective optimal hacking fidelities in both the symmetric and asymmetric problems is solely an extra multiplicative factor of $1/d_0^2$.


\begin{thebibliography}{47}%
\makeatletter
\providecommand \@ifxundefined [1]{%
 \@ifx{#1\undefined}
}%
\providecommand \@ifnum [1]{%
 \ifnum #1\expandafter \@firstoftwo
 \else \expandafter \@secondoftwo
 \fi
}%
\providecommand \@ifx [1]{%
 \ifx #1\expandafter \@firstoftwo
 \else \expandafter \@secondoftwo
 \fi
}%
\providecommand \natexlab [1]{#1}%
\providecommand \enquote  [1]{``#1''}%
\providecommand \bibnamefont  [1]{#1}%
\providecommand \bibfnamefont [1]{#1}%
\providecommand \citenamefont [1]{#1}%
\providecommand \href@noop [0]{\@secondoftwo}%
\providecommand \href [0]{\begingroup \@sanitize@url \@href}%
\providecommand \@href[1]{\@@startlink{#1}\@@href}%
\providecommand \@@href[1]{\endgroup#1\@@endlink}%
\providecommand \@sanitize@url [0]{\catcode `\\12\catcode `\$12\catcode
  `\&12\catcode `\#12\catcode `\^12\catcode `\_12\catcode `\%12\relax}%
\providecommand \@@startlink[1]{}%
\providecommand \@@endlink[0]{}%
\providecommand \url  [0]{\begingroup\@sanitize@url \@url }%
\providecommand \@url [1]{\endgroup\@href {#1}{\urlprefix }}%
\providecommand \urlprefix  [0]{URL }%
\providecommand \Eprint [0]{\href }%
\providecommand \doibase [0]{https://doi.org/}%
\providecommand \selectlanguage [0]{\@gobble}%
\providecommand \bibinfo  [0]{\@secondoftwo}%
\providecommand \bibfield  [0]{\@secondoftwo}%
\providecommand \translation [1]{[#1]}%
\providecommand \BibitemOpen [0]{}%
\providecommand \bibitemStop [0]{}%
\providecommand \bibitemNoStop [0]{.\EOS\space}%
\providecommand \EOS [0]{\spacefactor3000\relax}%
\providecommand \BibitemShut  [1]{\csname bibitem#1\endcsname}%
\let\auto@bib@innerbib\@empty
%</preamble>
\bibitem [{\citenamefont {Yurke}\ and\ \citenamefont
  {Denker}(1984)}]{yurke1984quantum}%
  \BibitemOpen
  \bibfield  {author} {\bibinfo {author} {\bibfnamefont {B.}~\bibnamefont
  {Yurke}}\ and\ \bibinfo {author} {\bibfnamefont {J.~S.}\ \bibnamefont
  {Denker}},\ }\bibfield  {title} {\bibinfo {title} {Quantum network theory},\
  }\href@noop {} {\bibfield  {journal} {\bibinfo  {journal} {Physical Review
  A}\ }\textbf {\bibinfo {volume} {29}},\ \bibinfo {pages} {1419} (\bibinfo
  {year} {1984})}\BibitemShut {NoStop}%
\bibitem [{\citenamefont {Elliott}(2002)}]{elliott2002building}%
  \BibitemOpen
  \bibfield  {author} {\bibinfo {author} {\bibfnamefont {C.}~\bibnamefont
  {Elliott}},\ }\bibfield  {title} {\bibinfo {title} {Building the quantum
  network},\ }\href@noop {} {\bibfield  {journal} {\bibinfo  {journal} {New
  Journal of Physics}\ }\textbf {\bibinfo {volume} {4}},\ \bibinfo {pages} {46}
  (\bibinfo {year} {2002})}\BibitemShut {NoStop}%
\bibitem [{\citenamefont {Simon}(2017)}]{simon2017towards}%
  \BibitemOpen
  \bibfield  {author} {\bibinfo {author} {\bibfnamefont {C.}~\bibnamefont
  {Simon}},\ }\bibfield  {title} {\bibinfo {title} {Towards a global quantum
  network},\ }\href@noop {} {\bibfield  {journal} {\bibinfo  {journal} {Nature
  Photonics}\ }\textbf {\bibinfo {volume} {11}},\ \bibinfo {pages} {678}
  (\bibinfo {year} {2017})}\BibitemShut {NoStop}%
\bibitem [{\citenamefont {Kimble}(2008)}]{kimble2008quantum}%
  \BibitemOpen
  \bibfield  {author} {\bibinfo {author} {\bibfnamefont {H.~J.}\ \bibnamefont
  {Kimble}},\ }\bibfield  {title} {\bibinfo {title} {The quantum internet},\
  }\href@noop {} {\bibfield  {journal} {\bibinfo  {journal} {Nature}\ }\textbf
  {\bibinfo {volume} {453}},\ \bibinfo {pages} {1023} (\bibinfo {year}
  {2008})}\BibitemShut {NoStop}%
\bibitem [{\citenamefont {Wehner}\ \emph {et~al.}(2018)\citenamefont {Wehner},
  \citenamefont {Elkouss},\ and\ \citenamefont {Hanson}}]{wehner2018quantum}%
  \BibitemOpen
  \bibfield  {author} {\bibinfo {author} {\bibfnamefont {S.}~\bibnamefont
  {Wehner}}, \bibinfo {author} {\bibfnamefont {D.}~\bibnamefont {Elkouss}},\
  and\ \bibinfo {author} {\bibfnamefont {R.}~\bibnamefont {Hanson}},\
  }\bibfield  {title} {\bibinfo {title} {Quantum internet: A vision for the
  road ahead},\ }\href@noop {} {\bibfield  {journal} {\bibinfo  {journal}
  {Science}\ }\textbf {\bibinfo {volume} {362}} (\bibinfo {year}
  {2018})}\BibitemShut {NoStop}%
\bibitem [{\citenamefont {Pirandola}\ and\ \citenamefont
  {Braunstein}(2016)}]{pirandola2016physics}%
  \BibitemOpen
  \bibfield  {author} {\bibinfo {author} {\bibfnamefont {S.}~\bibnamefont
  {Pirandola}}\ and\ \bibinfo {author} {\bibfnamefont {S.~L.}\ \bibnamefont
  {Braunstein}},\ }\bibfield  {title} {\bibinfo {title} {Physics: Unite to
  build a quantum internet},\ }\href@noop {} {\bibfield  {journal} {\bibinfo
  {journal} {Nature News}\ }\textbf {\bibinfo {volume} {532}},\ \bibinfo
  {pages} {169} (\bibinfo {year} {2016})}\BibitemShut {NoStop}%
\bibitem [{\citenamefont {Caleffi}\ \emph {et~al.}(2018)\citenamefont
  {Caleffi}, \citenamefont {Cacciapuoti},\ and\ \citenamefont
  {Bianchi}}]{caleffi2018quantum}%
  \BibitemOpen
  \bibfield  {author} {\bibinfo {author} {\bibfnamefont {M.}~\bibnamefont
  {Caleffi}}, \bibinfo {author} {\bibfnamefont {A.~S.}\ \bibnamefont
  {Cacciapuoti}},\ and\ \bibinfo {author} {\bibfnamefont {G.}~\bibnamefont
  {Bianchi}},\ }\bibfield  {title} {\bibinfo {title} {Quantum internet: From
  communication to distributed computing!},\ }in\ \href@noop {} {\emph
  {\bibinfo {booktitle} {Proceedings of the 5th ACM International Conference on
  Nanoscale Computing and Communication}}}\ (\bibinfo {year} {2018})\ pp.\
  \bibinfo {pages} {1--4}\BibitemShut {NoStop}%
\bibitem [{\citenamefont {Castelvecchi}(2018)}]{castelvecchi2018quantum}%
  \BibitemOpen
  \bibfield  {author} {\bibinfo {author} {\bibfnamefont {D.}~\bibnamefont
  {Castelvecchi}},\ }\bibfield  {title} {\bibinfo {title} {The quantum internet
  has arrived (and it hasn't)},\ }\href@noop {} {\bibfield  {journal} {\bibinfo
   {journal} {Nature}\ }\textbf {\bibinfo {volume} {554}} (\bibinfo {year}
  {2018})}\BibitemShut {NoStop}%
\bibitem [{\citenamefont {Sekino}\ and\ \citenamefont
  {Susskind}(2008)}]{sekino2008fast}%
  \BibitemOpen
  \bibfield  {author} {\bibinfo {author} {\bibfnamefont {Y.}~\bibnamefont
  {Sekino}}\ and\ \bibinfo {author} {\bibfnamefont {L.}~\bibnamefont
  {Susskind}},\ }\bibfield  {title} {\bibinfo {title} {Fast scramblers},\
  }\href@noop {} {\bibfield  {journal} {\bibinfo  {journal} {Journal of High
  Energy Physics}\ }\textbf {\bibinfo {volume} {2008}},\ \bibinfo {pages} {065}
  (\bibinfo {year} {2008})}\BibitemShut {NoStop}%
\bibitem [{\citenamefont {Hosur}\ \emph {et~al.}(2016)\citenamefont {Hosur},
  \citenamefont {Qi}, \citenamefont {Roberts},\ and\ \citenamefont
  {Yoshida}}]{hosur2016chaos}%
  \BibitemOpen
  \bibfield  {author} {\bibinfo {author} {\bibfnamefont {P.}~\bibnamefont
  {Hosur}}, \bibinfo {author} {\bibfnamefont {X.-L.}\ \bibnamefont {Qi}},
  \bibinfo {author} {\bibfnamefont {D.~A.}\ \bibnamefont {Roberts}},\ and\
  \bibinfo {author} {\bibfnamefont {B.}~\bibnamefont {Yoshida}},\ }\bibfield
  {title} {\bibinfo {title} {Chaos in quantum channels},\ }\href
  {https://doi.org/10.1007/JHEP02(2016)004} {\bibfield  {journal} {\bibinfo
  {journal} {Journal of High Energy Physics}\ }\textbf {\bibinfo {volume}
  {2016}},\ \bibinfo {pages} {4} (\bibinfo {year} {2016})}\BibitemShut
  {NoStop}%
\bibitem [{\citenamefont {Blok}\ \emph {et~al.}(2021)\citenamefont {Blok},
  \citenamefont {Ramasesh}, \citenamefont {Schuster}, \citenamefont {O'Brien},
  \citenamefont {Kreikebaum}, \citenamefont {Dahlen}, \citenamefont {Morvan},
  \citenamefont {Yoshida}, \citenamefont {Yao},\ and\ \citenamefont
  {Siddiqi}}]{blok2021scrambling}%
  \BibitemOpen
  \bibfield  {author} {\bibinfo {author} {\bibfnamefont {M.~S.}\ \bibnamefont
  {Blok}}, \bibinfo {author} {\bibfnamefont {V.~V.}\ \bibnamefont {Ramasesh}},
  \bibinfo {author} {\bibfnamefont {T.}~\bibnamefont {Schuster}}, \bibinfo
  {author} {\bibfnamefont {K.}~\bibnamefont {O'Brien}}, \bibinfo {author}
  {\bibfnamefont {J.~M.}\ \bibnamefont {Kreikebaum}}, \bibinfo {author}
  {\bibfnamefont {D.}~\bibnamefont {Dahlen}}, \bibinfo {author} {\bibfnamefont
  {A.}~\bibnamefont {Morvan}}, \bibinfo {author} {\bibfnamefont
  {B.}~\bibnamefont {Yoshida}}, \bibinfo {author} {\bibfnamefont {N.~Y.}\
  \bibnamefont {Yao}},\ and\ \bibinfo {author} {\bibfnamefont {I.}~\bibnamefont
  {Siddiqi}},\ }\bibfield  {title} {\bibinfo {title} {Quantum information
  scrambling on a superconducting qutrit processor},\ }\href
  {https://doi.org/10.1103/PhysRevX.11.021010} {\bibfield  {journal} {\bibinfo
  {journal} {Phys. Rev. X}\ }\textbf {\bibinfo {volume} {11}},\ \bibinfo
  {pages} {021010} (\bibinfo {year} {2021})}\BibitemShut {NoStop}%
\bibitem [{\citenamefont {Page}(1994)}]{page1994black}%
  \BibitemOpen
  \bibfield  {author} {\bibinfo {author} {\bibfnamefont {D.~N.}\ \bibnamefont
  {Page}},\ }\bibfield  {title} {\bibinfo {title} {Black hole information},\
  }in\ \href@noop {} {\emph {\bibinfo {booktitle} {Proceedings of the 5th
  Canadian conference on general relativity and relativistic astrophysics}}},\
  Vol.~\bibinfo {volume} {1}\ (\bibinfo {organization} {World Scientific},\
  \bibinfo {year} {1994})\ pp.\ \bibinfo {pages} {1--41}\BibitemShut {NoStop}%
\bibitem [{\citenamefont {Hayden}\ and\ \citenamefont
  {Preskill}(2007)}]{hayden2007black}%
  \BibitemOpen
  \bibfield  {author} {\bibinfo {author} {\bibfnamefont {P.}~\bibnamefont
  {Hayden}}\ and\ \bibinfo {author} {\bibfnamefont {J.}~\bibnamefont
  {Preskill}},\ }\bibfield  {title} {\bibinfo {title} {Black holes as mirrors:
  quantum information in random subsystems},\ }\href@noop {} {\bibfield
  {journal} {\bibinfo  {journal} {Journal of high energy physics}\ }\textbf
  {\bibinfo {volume} {2007}},\ \bibinfo {pages} {120} (\bibinfo {year}
  {2007})}\BibitemShut {NoStop}%
\bibitem [{\citenamefont {Bao}\ and\ \citenamefont
  {Kikuchi}(2021)}]{bao2021hayden}%
  \BibitemOpen
  \bibfield  {author} {\bibinfo {author} {\bibfnamefont {N.}~\bibnamefont
  {Bao}}\ and\ \bibinfo {author} {\bibfnamefont {Y.}~\bibnamefont {Kikuchi}},\
  }\bibfield  {title} {\bibinfo {title} {Hayden-preskill decoding from noisy
  hawking radiation},\ }\href@noop {} {\bibfield  {journal} {\bibinfo
  {journal} {Journal of High Energy Physics}\ }\textbf {\bibinfo {volume}
  {2021}},\ \bibinfo {pages} {1} (\bibinfo {year} {2021})}\BibitemShut
  {NoStop}%
\bibitem [{\citenamefont {Cheng}\ \emph {et~al.}(2020)\citenamefont {Cheng},
  \citenamefont {Liu}, \citenamefont {Guo}, \citenamefont {Chen}, \citenamefont
  {Zhang},\ and\ \citenamefont {Zhai}}]{cheng2020realizing}%
  \BibitemOpen
  \bibfield  {author} {\bibinfo {author} {\bibfnamefont {Y.}~\bibnamefont
  {Cheng}}, \bibinfo {author} {\bibfnamefont {C.}~\bibnamefont {Liu}}, \bibinfo
  {author} {\bibfnamefont {J.}~\bibnamefont {Guo}}, \bibinfo {author}
  {\bibfnamefont {Y.}~\bibnamefont {Chen}}, \bibinfo {author} {\bibfnamefont
  {P.}~\bibnamefont {Zhang}},\ and\ \bibinfo {author} {\bibfnamefont
  {H.}~\bibnamefont {Zhai}},\ }\bibfield  {title} {\bibinfo {title} {Realizing
  the hayden-preskill protocol with coupled dicke models},\ }\href@noop {}
  {\bibfield  {journal} {\bibinfo  {journal} {Physical Review Research}\
  }\textbf {\bibinfo {volume} {2}},\ \bibinfo {pages} {043024} (\bibinfo {year}
  {2020})}\BibitemShut {NoStop}%
\bibitem [{\citenamefont {Yoshida}(2019)}]{yoshida2019firewalls}%
  \BibitemOpen
  \bibfield  {author} {\bibinfo {author} {\bibfnamefont {B.}~\bibnamefont
  {Yoshida}},\ }\bibfield  {title} {\bibinfo {title} {Firewalls vs.
  scrambling},\ }\href@noop {} {\bibfield  {journal} {\bibinfo  {journal}
  {Journal of High Energy Physics}\ }\textbf {\bibinfo {volume} {2019}},\
  \bibinfo {pages} {1} (\bibinfo {year} {2019})}\BibitemShut {NoStop}%
\bibitem [{\citenamefont {Yoshida}\ and\ \citenamefont
  {Kitaev}(2017)}]{yoshida2017efficient}%
  \BibitemOpen
  \bibfield  {author} {\bibinfo {author} {\bibfnamefont {B.}~\bibnamefont
  {Yoshida}}\ and\ \bibinfo {author} {\bibfnamefont {A.}~\bibnamefont
  {Kitaev}},\ }\bibfield  {title} {\bibinfo {title} {Efficient decoding for the
  hayden-preskill protocol},\ }\href@noop {} {\bibfield  {journal} {\bibinfo
  {journal} {arXiv preprint arXiv:1710.03363}\ } (\bibinfo {year}
  {2017})}\BibitemShut {NoStop}%
\bibitem [{\citenamefont {Dupuis}(2010)}]{dupuis2010decoupling}%
  \BibitemOpen
  \bibfield  {author} {\bibinfo {author} {\bibfnamefont {F.}~\bibnamefont
  {Dupuis}},\ }\bibfield  {title} {\bibinfo {title} {The decoupling approach to
  quantum information theory},\ }\href@noop {} {\bibfield  {journal} {\bibinfo
  {journal} {arXiv preprint arXiv:1004.1641}\ } (\bibinfo {year}
  {2010})}\BibitemShut {NoStop}%
\bibitem [{\citenamefont {Boes}\ \emph {et~al.}(2018)\citenamefont {Boes},
  \citenamefont {Wilming}, \citenamefont {Gallego},\ and\ \citenamefont
  {Eisert}}]{boes2018catalytic}%
  \BibitemOpen
  \bibfield  {author} {\bibinfo {author} {\bibfnamefont {P.}~\bibnamefont
  {Boes}}, \bibinfo {author} {\bibfnamefont {H.}~\bibnamefont {Wilming}},
  \bibinfo {author} {\bibfnamefont {R.}~\bibnamefont {Gallego}},\ and\ \bibinfo
  {author} {\bibfnamefont {J.}~\bibnamefont {Eisert}},\ }\bibfield  {title}
  {\bibinfo {title} {Catalytic quantum randomness},\ }\href@noop {} {\bibfield
  {journal} {\bibinfo  {journal} {Physical Review X}\ }\textbf {\bibinfo
  {volume} {8}},\ \bibinfo {pages} {041016} (\bibinfo {year}
  {2018})}\BibitemShut {NoStop}%
\bibitem [{\citenamefont {Wilming}(2020)}]{wilming2020entropy}%
  \BibitemOpen
  \bibfield  {author} {\bibinfo {author} {\bibfnamefont {H.}~\bibnamefont
  {Wilming}},\ }\bibfield  {title} {\bibinfo {title} {Entropy and reversible
  catalysis},\ }\href@noop {} {\bibfield  {journal} {\bibinfo  {journal} {arXiv
  preprint arXiv:2012.05573}\ } (\bibinfo {year} {2020})}\BibitemShut {NoStop}%
\bibitem [{\citenamefont {Lie}\ \emph {et~al.}(2019)\citenamefont {Lie},
  \citenamefont {Kwon}, \citenamefont {Kim},\ and\ \citenamefont
  {Jeong}}]{lie2019unconditionally}%
  \BibitemOpen
  \bibfield  {author} {\bibinfo {author} {\bibfnamefont {S.~H.}\ \bibnamefont
  {Lie}}, \bibinfo {author} {\bibfnamefont {H.}~\bibnamefont {Kwon}}, \bibinfo
  {author} {\bibfnamefont {M.}~\bibnamefont {Kim}},\ and\ \bibinfo {author}
  {\bibfnamefont {H.}~\bibnamefont {Jeong}},\ }\bibfield  {title} {\bibinfo
  {title} {Unconditionally secure qubit commitment scheme using quantum
  maskers},\ }\href@noop {} {\bibfield  {journal} {\bibinfo  {journal} {arXiv
  preprint arXiv:1903.12304}\ } (\bibinfo {year} {2019})}\BibitemShut {NoStop}%
\bibitem [{\citenamefont {Lie}\ and\ \citenamefont
  {Jeong}(2020{\natexlab{a}})}]{lie2020randomness}%
  \BibitemOpen
  \bibfield  {author} {\bibinfo {author} {\bibfnamefont {S.~H.}\ \bibnamefont
  {Lie}}\ and\ \bibinfo {author} {\bibfnamefont {H.}~\bibnamefont {Jeong}},\
  }\bibfield  {title} {\bibinfo {title} {Randomness cost of masking quantum
  information and the information conservation law},\ }\href@noop {} {\bibfield
   {journal} {\bibinfo  {journal} {Physical Review A}\ }\textbf {\bibinfo
  {volume} {101}},\ \bibinfo {pages} {052322} (\bibinfo {year}
  {2020}{\natexlab{a}})}\BibitemShut {NoStop}%
\bibitem [{\citenamefont {Lie}\ and\ \citenamefont
  {Jeong}(2020{\natexlab{b}})}]{lie2020uniform}%
  \BibitemOpen
  \bibfield  {author} {\bibinfo {author} {\bibfnamefont {S.~H.}\ \bibnamefont
  {Lie}}\ and\ \bibinfo {author} {\bibfnamefont {H.}~\bibnamefont {Jeong}},\
  }\bibfield  {title} {\bibinfo {title} {Only uniform randomness can yield
  quantum advantages},\ }\href@noop {} {\bibfield  {journal} {\bibinfo
  {journal} {arXiv preprint arXiv:2010.14795}\ } (\bibinfo {year}
  {2020}{\natexlab{b}})}\BibitemShut {NoStop}%
\bibitem [{\citenamefont {Wootters}\ and\ \citenamefont
  {Zurek}(1982)}]{wootters1982single}%
  \BibitemOpen
  \bibfield  {author} {\bibinfo {author} {\bibfnamefont {W.~K.}\ \bibnamefont
  {Wootters}}\ and\ \bibinfo {author} {\bibfnamefont {W.~H.}\ \bibnamefont
  {Zurek}},\ }\bibfield  {title} {\bibinfo {title} {A single quantum cannot be
  cloned},\ }\href@noop {} {\bibfield  {journal} {\bibinfo  {journal} {Nature}\
  }\textbf {\bibinfo {volume} {299}},\ \bibinfo {pages} {802} (\bibinfo {year}
  {1982})}\BibitemShut {NoStop}%
\bibitem [{\citenamefont {Garcia-Escartin}\ and\ \citenamefont
  {Chamorro-Posada}(2013)}]{garcia2013swap}%
  \BibitemOpen
  \bibfield  {author} {\bibinfo {author} {\bibfnamefont {J.~C.}\ \bibnamefont
  {Garcia-Escartin}}\ and\ \bibinfo {author} {\bibfnamefont {P.}~\bibnamefont
  {Chamorro-Posada}},\ }\bibfield  {title} {\bibinfo {title} {A swap gate for
  qudits},\ }\href@noop {} {\bibfield  {journal} {\bibinfo  {journal} {Quantum
  information processing}\ }\textbf {\bibinfo {volume} {12}},\ \bibinfo {pages}
  {3625} (\bibinfo {year} {2013})}\BibitemShut {NoStop}%
\bibitem [{SM_()}]{SM_Qhack}%
  \BibitemOpen
  \href@noop {} {\bibinfo {title} {See {S}upplemental {M}aterial for
  arguments against \emph{arbitrary} quantum-state installation, derivations of
  hacking-fidelity bounds, numerical optimization for quantum hacking and
  derivations of asymptotic optimal-hacking fidelity formulas}}\BibitemShut
  {NoStop}%
\bibitem [{\citenamefont {Horodecki}\ \emph {et~al.}(2005)\citenamefont
  {Horodecki}, \citenamefont {Oppenheim},\ and\ \citenamefont
  {Winter}}]{horodecki2005partial}%
  \BibitemOpen
  \bibfield  {author} {\bibinfo {author} {\bibfnamefont {M.}~\bibnamefont
  {Horodecki}}, \bibinfo {author} {\bibfnamefont {J.}~\bibnamefont
  {Oppenheim}},\ and\ \bibinfo {author} {\bibfnamefont {A.}~\bibnamefont
  {Winter}},\ }\bibfield  {title} {\bibinfo {title} {Partial quantum
  information},\ }\href@noop {} {\bibfield  {journal} {\bibinfo  {journal}
  {Nature}\ }\textbf {\bibinfo {volume} {436}},\ \bibinfo {pages} {673}
  (\bibinfo {year} {2005})}\BibitemShut {NoStop}%
\bibitem [{\citenamefont {Horodecki}\ \emph {et~al.}(1999)\citenamefont
  {Horodecki}, \citenamefont {Horodecki},\ and\ \citenamefont
  {Horodecki}}]{horodecki1999general}%
  \BibitemOpen
  \bibfield  {author} {\bibinfo {author} {\bibfnamefont {M.}~\bibnamefont
  {Horodecki}}, \bibinfo {author} {\bibfnamefont {P.}~\bibnamefont
  {Horodecki}},\ and\ \bibinfo {author} {\bibfnamefont {R.}~\bibnamefont
  {Horodecki}},\ }\bibfield  {title} {\bibinfo {title} {General teleportation
  channel, singlet fraction, and quasidistillation},\ }\href@noop {} {\bibfield
   {journal} {\bibinfo  {journal} {Physical Review A}\ }\textbf {\bibinfo
  {volume} {60}},\ \bibinfo {pages} {1888} (\bibinfo {year}
  {1999})}\BibitemShut {NoStop}%
\bibitem [{\citenamefont {Montangero}\ \emph {et~al.}(2018)\citenamefont
  {Montangero}, \citenamefont {Montangero},\ and\ \citenamefont
  {Evenson}}]{montangero2018introduction}%
  \BibitemOpen
  \bibfield  {author} {\bibinfo {author} {\bibfnamefont {S.}~\bibnamefont
  {Montangero}}, \bibinfo {author} {\bibnamefont {Montangero}},\ and\ \bibinfo
  {author} {\bibnamefont {Evenson}},\ }\href@noop {} {\emph {\bibinfo {title}
  {Introduction to Tensor Network Methods}}}\ (\bibinfo  {publisher}
  {Springer},\ \bibinfo {year} {2018})\BibitemShut {NoStop}%
\bibitem [{\citenamefont {Landsberg}\ \emph {et~al.}(2011)\citenamefont
  {Landsberg}, \citenamefont {Qi},\ and\ \citenamefont
  {Ye}}]{landsberg2011geometry}%
  \BibitemOpen
  \bibfield  {author} {\bibinfo {author} {\bibfnamefont {J.~M.}\ \bibnamefont
  {Landsberg}}, \bibinfo {author} {\bibfnamefont {Y.}~\bibnamefont {Qi}},\ and\
  \bibinfo {author} {\bibfnamefont {K.}~\bibnamefont {Ye}},\ }\bibfield
  {title} {\bibinfo {title} {On the geometry of tensor network states},\
  }\href@noop {} {\bibfield  {journal} {\bibinfo  {journal} {arXiv preprint
  arXiv:1105.4449}\ } (\bibinfo {year} {2011})}\BibitemShut {NoStop}%
\bibitem [{\citenamefont {{\.Z}yczkowski}\ and\ \citenamefont
  {Bengtsson}(2004)}]{zyczkowski2004duality}%
  \BibitemOpen
  \bibfield  {author} {\bibinfo {author} {\bibfnamefont {K.}~\bibnamefont
  {{\.Z}yczkowski}}\ and\ \bibinfo {author} {\bibfnamefont {I.}~\bibnamefont
  {Bengtsson}},\ }\bibfield  {title} {\bibinfo {title} {On duality between
  quantum maps and quantum states},\ }\href@noop {} {\bibfield  {journal}
  {\bibinfo  {journal} {Open systems \& information dynamics}\ }\textbf
  {\bibinfo {volume} {11}},\ \bibinfo {pages} {3} (\bibinfo {year}
  {2004})}\BibitemShut {NoStop}%
\bibitem [{\citenamefont {Miszczak}(2011)}]{miszczak2011singular}%
  \BibitemOpen
  \bibfield  {author} {\bibinfo {author} {\bibfnamefont {J.~A.}\ \bibnamefont
  {Miszczak}},\ }\bibfield  {title} {\bibinfo {title} {Singular value
  decomposition and matrix reorderings in quantum information theory},\
  }\href@noop {} {\bibfield  {journal} {\bibinfo  {journal} {International
  Journal of Modern Physics C}\ }\textbf {\bibinfo {volume} {22}},\ \bibinfo
  {pages} {897} (\bibinfo {year} {2011})}\BibitemShut {NoStop}%
\bibitem [{\citenamefont {Bruzda}\ \emph {et~al.}(2009)\citenamefont {Bruzda},
  \citenamefont {Cappellini}, \citenamefont {Sommers},\ and\ \citenamefont
  {{\.Z}yczkowski}}]{bruzda2009random}%
  \BibitemOpen
  \bibfield  {author} {\bibinfo {author} {\bibfnamefont {W.}~\bibnamefont
  {Bruzda}}, \bibinfo {author} {\bibfnamefont {V.}~\bibnamefont {Cappellini}},
  \bibinfo {author} {\bibfnamefont {H.-J.}\ \bibnamefont {Sommers}},\ and\
  \bibinfo {author} {\bibfnamefont {K.}~\bibnamefont {{\.Z}yczkowski}},\
  }\bibfield  {title} {\bibinfo {title} {Random quantum operations},\
  }\href@noop {} {\bibfield  {journal} {\bibinfo  {journal} {Physics Letters
  A}\ }\textbf {\bibinfo {volume} {373}},\ \bibinfo {pages} {320} (\bibinfo
  {year} {2009})}\BibitemShut {NoStop}%
\bibitem [{1()}]{1}%
  \BibitemOpen
  \href@noop {} {}\bibinfo {note} {Although $R$ is $d_B^2\times d_B^2$, it only
  acts on a $d_A^2$-dimensional subspace $\Im \{(I_B \otimes \chi )U^o\}$ to
  the right and $(\mathrm {Ker } \{U^o\})^\perp $ to the left, so that we may
  treat $R$ either as a rank-$d_A^2$ partial unitary matrix or a $d_A^2 \times
  d_B^2$ coisometry without losing generality.}\BibitemShut {Stop}%
\bibitem [{\citenamefont {Lesovik}\ \emph {et~al.}(2019)\citenamefont
  {Lesovik}, \citenamefont {Sadovskyy}, \citenamefont {Suslov}, \citenamefont
  {Lebedev},\ and\ \citenamefont {Vinokur}}]{lesovik2019arrow}%
  \BibitemOpen
  \bibfield  {author} {\bibinfo {author} {\bibfnamefont {G.~B.}\ \bibnamefont
  {Lesovik}}, \bibinfo {author} {\bibfnamefont {I.~A.}\ \bibnamefont
  {Sadovskyy}}, \bibinfo {author} {\bibfnamefont {M.}~\bibnamefont {Suslov}},
  \bibinfo {author} {\bibfnamefont {A.~V.}\ \bibnamefont {Lebedev}},\ and\
  \bibinfo {author} {\bibfnamefont {V.~M.}\ \bibnamefont {Vinokur}},\
  }\bibfield  {title} {\bibinfo {title} {Arrow of time and its reversal on the
  ibm quantum computer},\ }\href@noop {} {\bibfield  {journal} {\bibinfo
  {journal} {Scientific reports}\ }\textbf {\bibinfo {volume} {9}},\ \bibinfo
  {pages} {1} (\bibinfo {year} {2019})}\BibitemShut {NoStop}%
\bibitem [{\citenamefont {Teo}\ \emph {et~al.}(2011{\natexlab{a}})\citenamefont
  {Teo}, \citenamefont {Zhu}, \citenamefont {Englert}, \citenamefont {\ifmmode
  \check{R}\else \v{R}\fi{}eh\'a\ifmmode~\check{c}\else \v{c}\fi{}ek},\ and\
  \citenamefont {Hradil}}]{teo2011mlme}%
  \BibitemOpen
  \bibfield  {author} {\bibinfo {author} {\bibfnamefont {Y.~S.}\ \bibnamefont
  {Teo}}, \bibinfo {author} {\bibfnamefont {H.}~\bibnamefont {Zhu}}, \bibinfo
  {author} {\bibfnamefont {B.-G.}\ \bibnamefont {Englert}}, \bibinfo {author}
  {\bibfnamefont {J.}~\bibnamefont {\ifmmode \check{R}\else
  \v{R}\fi{}eh\'a\ifmmode~\check{c}\else \v{c}\fi{}ek}},\ and\ \bibinfo
  {author} {\bibfnamefont {Z.~c.~v.}\ \bibnamefont {Hradil}},\ }\bibfield
  {title} {\bibinfo {title} {Quantum-state reconstruction by maximizing
  likelihood and entropy},\ }\href
  {https://doi.org/10.1103/PhysRevLett.107.020404} {\bibfield  {journal}
  {\bibinfo  {journal} {Phys. Rev. Lett.}\ }\textbf {\bibinfo {volume} {107}},\
  \bibinfo {pages} {020404} (\bibinfo {year} {2011}{\natexlab{a}})}\BibitemShut
  {NoStop}%
\bibitem [{\citenamefont {Teo}\ \emph {et~al.}(2011{\natexlab{b}})\citenamefont
  {Teo}, \citenamefont {Englert}, \citenamefont {\ifmmode \check{R}\else
  \v{R}\fi{}eh\'a\ifmmode~\check{c}\else \v{c}\fi{}ek},\ and\ \citenamefont
  {Hradil}}]{teo2011adaptive}%
  \BibitemOpen
  \bibfield  {author} {\bibinfo {author} {\bibfnamefont {Y.~S.}\ \bibnamefont
  {Teo}}, \bibinfo {author} {\bibfnamefont {B.-G.}\ \bibnamefont {Englert}},
  \bibinfo {author} {\bibfnamefont {J.}~\bibnamefont {\ifmmode \check{R}\else
  \v{R}\fi{}eh\'a\ifmmode~\check{c}\else \v{c}\fi{}ek}},\ and\ \bibinfo
  {author} {\bibfnamefont {Z.~c.~v.}\ \bibnamefont {Hradil}},\ }\bibfield
  {title} {\bibinfo {title} {Adaptive schemes for incomplete quantum process
  tomography},\ }\href {https://doi.org/10.1103/PhysRevA.84.062125} {\bibfield
  {journal} {\bibinfo  {journal} {Phys. Rev. A}\ }\textbf {\bibinfo {volume}
  {84}},\ \bibinfo {pages} {062125} (\bibinfo {year}
  {2011}{\natexlab{b}})}\BibitemShut {NoStop}%
\bibitem [{\citenamefont {Page}(1993)}]{page1993entropy}%
  \BibitemOpen
  \bibfield  {author} {\bibinfo {author} {\bibfnamefont {D.~N.}\ \bibnamefont
  {Page}},\ }\bibfield  {title} {\bibinfo {title} {Average entropy of a
  subsystem},\ }\href {https://doi.org/10.1103/PhysRevLett.71.1291} {\bibfield
  {journal} {\bibinfo  {journal} {Phys. Rev. Lett.}\ }\textbf {\bibinfo
  {volume} {71}},\ \bibinfo {pages} {1291} (\bibinfo {year}
  {1993})}\BibitemShut {NoStop}%
\bibitem [{\citenamefont {Mehta}(2004)}]{mehta2004random}%
  \BibitemOpen
  \bibfield  {author} {\bibinfo {author} {\bibfnamefont {M.~L.}\ \bibnamefont
  {Mehta}},\ }\href@noop {} {\emph {\bibinfo {title} {Random Matrices}}}\
  (\bibinfo  {publisher} {Elsevier},\ \bibinfo {address} {Amsterdam},\ \bibinfo
  {year} {2004})\BibitemShut {NoStop}%
\bibitem [{\citenamefont {Mar{\v c}enko}\ and\ \citenamefont
  {Pastur}(1967)}]{marcenko1967eigenvalues}%
  \BibitemOpen
  \bibfield  {author} {\bibinfo {author} {\bibfnamefont {V.~A.}\ \bibnamefont
  {Mar{\v c}enko}}\ and\ \bibinfo {author} {\bibfnamefont {L.~A.}\ \bibnamefont
  {Pastur}},\ }\bibfield  {title} {\bibinfo {title} {Quantum mutual information
  and the one-time pad},\ }\href@noop {} {\bibfield  {journal} {\bibinfo
  {journal} {Math. USSR Sb.}\ }\textbf {\bibinfo {volume} {1}},\ \bibinfo
  {pages} {457} (\bibinfo {year} {1967})}\BibitemShut {NoStop}%
\bibitem [{SpF(2014)}]{SpFuncBk}%
  \BibitemOpen
  \bibfield  {title} {\bibinfo {title} {{S}pecial {F}unctions},\ }in\ \href
  {https://doi.org/https://doi.org/10.1016/B978-0-12-384933-5.00009-6} {\emph
  {\bibinfo {booktitle} {Table of Integrals, Series, and Products (Eighth
  Edition)}}},\ \bibinfo {editor} {edited by\ \bibinfo {editor} {\bibfnamefont
  {D.}~\bibnamefont {Zwillinger}}, \bibinfo {editor} {\bibfnamefont
  {V.}~\bibnamefont {Moll}}, \bibinfo {editor} {\bibfnamefont {I.}~\bibnamefont
  {Gradshteyn}},\ and\ \bibinfo {editor} {\bibfnamefont {I.}~\bibnamefont
  {Ryzhik}}}\ (\bibinfo  {publisher} {Academic Press},\ \bibinfo {address}
  {Boston},\ \bibinfo {year} {2014})\ \bibinfo {edition} {eighth edition}\
  ed.,\ pp.\ \bibinfo {pages} {1014--1059}\BibitemShut {NoStop}%
\bibitem [{\citenamefont {Tajima}\ and\ \citenamefont
  {Saito}(2021)}]{tajima2021symmetry}%
  \BibitemOpen
  \bibfield  {author} {\bibinfo {author} {\bibfnamefont {H.}~\bibnamefont
  {Tajima}}\ and\ \bibinfo {author} {\bibfnamefont {K.}~\bibnamefont {Saito}},\
  }\bibfield  {title} {\bibinfo {title} {Symmetry hinders quantum information
  recovery},\ }\href@noop {} {\bibfield  {journal} {\bibinfo  {journal} {arXiv
  preprint arXiv:2103.01876}\ } (\bibinfo {year} {2021})}\BibitemShut {NoStop}%
\bibitem [{\citenamefont {Nielsen}\ and\ \citenamefont
  {Chuang}(1997)}]{nielsen1997program}%
  \BibitemOpen
  \bibfield  {author} {\bibinfo {author} {\bibfnamefont {M.~A.}\ \bibnamefont
  {Nielsen}}\ and\ \bibinfo {author} {\bibfnamefont {I.~L.}\ \bibnamefont
  {Chuang}},\ }\bibfield  {title} {\bibinfo {title} {Programmable quantum gate
  arrays},\ }\href {https://doi.org/10.1103/PhysRevLett.79.321} {\bibfield
  {journal} {\bibinfo  {journal} {Phys. Rev. Lett.}\ }\textbf {\bibinfo
  {volume} {79}},\ \bibinfo {pages} {321} (\bibinfo {year} {1997})}\BibitemShut
  {NoStop}%
\bibitem [{\citenamefont {Oshita}\ \emph {et~al.}(2020)\citenamefont {Oshita},
  \citenamefont {Wang},\ and\ \citenamefont
  {Afshordi}}]{oshita2020reflectivity}%
  \BibitemOpen
  \bibfield  {author} {\bibinfo {author} {\bibfnamefont {N.}~\bibnamefont
  {Oshita}}, \bibinfo {author} {\bibfnamefont {Q.}~\bibnamefont {Wang}},\ and\
  \bibinfo {author} {\bibfnamefont {N.}~\bibnamefont {Afshordi}},\ }\bibfield
  {title} {\bibinfo {title} {On reflectivity of quantum black hole horizons},\
  }\href@noop {} {\bibfield  {journal} {\bibinfo  {journal} {Journal of
  Cosmology and Astroparticle Physics}\ }\textbf {\bibinfo {volume}
  {2020}}\bibinfo  {number} { (04)},\ \bibinfo {pages} {016}}\BibitemShut
  {NoStop}%
\bibitem [{\citenamefont {Wang}\ \emph {et~al.}(2020)\citenamefont {Wang},
  \citenamefont {Oshita},\ and\ \citenamefont {Afshordi}}]{wang2020echoes}%
  \BibitemOpen
\bibfield  {number} {  }\bibfield  {author} {\bibinfo {author} {\bibfnamefont
  {Q.}~\bibnamefont {Wang}}, \bibinfo {author} {\bibfnamefont {N.}~\bibnamefont
  {Oshita}},\ and\ \bibinfo {author} {\bibfnamefont {N.}~\bibnamefont
  {Afshordi}},\ }\bibfield  {title} {\bibinfo {title} {Echoes from quantum
  black holes},\ }\href@noop {} {\bibfield  {journal} {\bibinfo  {journal}
  {Physical Review D}\ }\textbf {\bibinfo {volume} {101}},\ \bibinfo {pages}
  {024031} (\bibinfo {year} {2020})}\BibitemShut {NoStop}%
\bibitem [{\citenamefont {Uhlmann}(1976)}]{uhlmann1976transition}%
  \BibitemOpen
  \bibfield  {author} {\bibinfo {author} {\bibfnamefont {A.}~\bibnamefont
  {Uhlmann}},\ }\bibfield  {title} {\bibinfo {title} {The “transition
  probability” in the state space of a∗-algebra},\ }\href@noop {}
  {\bibfield  {journal} {\bibinfo  {journal} {Reports on Mathematical Physics}\
  }\textbf {\bibinfo {volume} {9}},\ \bibinfo {pages} {273} (\bibinfo {year}
  {1976})}\BibitemShut {NoStop}%
\bibitem [{\citenamefont {Fuchs}\ and\ \citenamefont {Van
  De~Graaf}(1999)}]{fuchs1999cryptographic}%
  \BibitemOpen
  \bibfield  {author} {\bibinfo {author} {\bibfnamefont {C.~A.}\ \bibnamefont
  {Fuchs}}\ and\ \bibinfo {author} {\bibfnamefont {J.}~\bibnamefont {Van
  De~Graaf}},\ }\bibfield  {title} {\bibinfo {title} {Cryptographic
  distinguishability measures for quantum-mechanical states},\ }\href@noop {}
  {\bibfield  {journal} {\bibinfo  {journal} {IEEE Transactions on Information
  Theory}\ }\textbf {\bibinfo {volume} {45}},\ \bibinfo {pages} {1216}
  (\bibinfo {year} {1999})}\BibitemShut {NoStop}%
\end{thebibliography}
\end{document}